\documentclass[a4paper,onecolumn,11pt,accepted=2024-07-27]{quantumarticle}
\pdfoutput=1
\usepackage{bm}
\usepackage{bbm}
\usepackage{mathrsfs}
\usepackage{amssymb}
\usepackage{amsmath}
\usepackage{amsfonts}
\usepackage{amsthm}
\usepackage[monochrome]{xcolor}
\usepackage[numbers]{natbib}
\usepackage{hyperref}
\usepackage{graphicx}
\usepackage[ruled, noline]{algorithm2e}
\usepackage{amssymb}
\usepackage{multirow}
\usepackage{enumerate}
\newcommand{\ave}[1]{\left\langle#1\right\rangle}

\newcommand{\ket}[1]{\left|#1\right\rangle}

\newcommand{\bra}[1]{\left\langle#1\right|}

\newcommand{\dketbra}[1]{|#1\rangle\langle#1|}
\newcommand{\be}{\begin{equation}}

\newcommand{\abs}[1]{\left|#1\right|}
\newcommand{\norm}[1]{\left|\left|#1\right|\right|}
\newcommand{\cS}{{\cal S}}

\newcommand{\cF}{{\cal F}}

\newcommand{\cN}{{\cal N}}

\newcommand{\cI}{{\cal I}}
\newcommand{\cL}{{\cal L}}
\newcommand{\cW}{{\cal W}}
\newcommand{\cC}{{\cal C}}
\newcommand{\cM}{{\cal M}}
\newcommand{\cG}{{\cal G}}
\newcommand{\cX}{{\cal X}}
\newcommand{\cY}{{\cal Y}}
\newcommand{\cZ}{{\cal Z}}
\newcommand{\cE}{{\cal E}}
\newcommand{\E}[2]{\mathbb{E}_{#1}#2}
\newcommand{\R}{\mathbb R}
\newcommand{\C}{\mathbb C}
\newcommand{\N}{\mathbb N}

\newcommand{\bx}{{\bf x}}
\newcommand{\by}{{\bf y}}

\newcommand{\bk}{{\bf k}}
\newcommand{\br}{{\bf r}}
\newcommand{\bs}{{\bf s}}

\newcommand{\bd}{{\bf d}}

\newcommand{\bfm}{{\bf m}}
\newcommand{\bxi}{{\boldsymbol \xi}}

\newcommand{\bgamma}{{\boldsymbol \gamma}}
\newcommand{\balpha}{{\boldsymbol \alpha}}
\renewcommand{\det}[1]{{\rm Det}\left(#1\right)}
\newcommand{\tr}[1]{{\rm Tr}\left[#1\right]}
\newcommand{\pdim}[1]{{\rm Pdim}(#1)}

\newcommand{\fdim}[2]{{\rm fat}_{#1}\left(#2\right)}

\renewcommand{\i}{{\rm i}}
\newcommand{\1}{\mathbbm{1}}

\newtheorem{lemma}{Lemma}
\newtheorem{theorem}{Theorem}
\newtheorem{definition}{Definition}
\newtheorem{corollary}{Corollary}

\begin{document}

\title{A learning theory for quantum photonic processors and beyond}
%\title{Learning complexity of Gaussian circuits and photodetection}
%\title{Quantum statistical learning in infinite dimensions}

\author{Matteo Rosati}
\address{Dipartimento di Ingegneria Civile, Informatica e delle Tecnologie Aeronautiche, Università Roma Tre, Via Vito Volterra 62, Rome, I-00146, Italy}

\begin{abstract}
We consider the tasks of learning quantum states, measurements and channels generated by continuous-variable (CV) quantum circuits. This family of circuits is suited to describe optical quantum technologies and in particular it includes state-of-the-art photonic processors capable of showing quantum advantage. We define classes of functions that map classical variables, encoded into the CV circuit parameters, to outcome probabilities evaluated on those circuits. We then establish efficient learnability guarantees for such classes, by computing bounds on their pseudo-dimension or covering numbers, showing that CV quantum circuits can be learned with a sample complexity that scales polynomially with the circuit's size, i.e., the number of modes. Our results show that CV circuits can be trained efficiently using a number of training samples that, unlike their finite-dimensional counterpart, does not scale with the circuit depth.
\end{abstract}

\maketitle
\tableofcontents

\section{Introduction}
The last years have seen an incredible advancement in hardware solutions for quantum technologies. In particular, the recent demonstration of a quantum computational advantage via photonic circuits~\cite{Zhong2020,Madsen2022} finally paves the way for the realization of full-fledged quantum information processing with light, a solution that bears intrinsic advantages with respect to other platforms, in terms of scalability, robustness and deployability~\cite{Hoch2021,EliBourassa2021,Wang2020}. At the same time, the increased control of infinite-dimensional quantum states in several other platforms, such as cavity~\cite{Krastanov2015,Eickbusch2021} or mechanical resonators~\cite{Arrangoiz-Arriola2019}, is pushing the boundaries of continuous-variable (CV) quantum information processing beyond photonics. Finally, the increased interplay between qubit and CV platforms~\cite{Andersen2014a,Hastrup2019} spurs the interest into the development of quantum error correction codes~\cite{Michael2016,Noh2020,Tosta2022a} and provides an alternative to more standard approaches for quantum technologies. A combination of the aforementioned events thus marks a renewed surge of interest into CV information processing. 

From a theoretical perspective, the characterization of the information-processing capabilities of quantum devices has been recently subject to a paradigm shift, thanks to the introduction of statistical learning techniques~\cite{Aaronson2007,Cheng2015,Sweke2020a,Caro2020a,Caro2021b}, which underly the success of classical machine learning~\cite{Vapnik1999,Vapnik2000,Anthony1999}. In this approach, one recognizes that a successful use of quantum devices often requires two ingredients: (i) the estimation of quantities of interest about the quantum states or processes running in the device; (ii) the optimization of the device's parameter setup based on the estimated data, in order to maximize the device's performance in a specific task. 
Therefore, in this setting one can characterize the information-processing capabilities of a quantum device by quantifying its sample complexity (SC), i.e., the number of repetitions of an experiment that are required to estimate specific quantities and/or find a good device setup for carrying out a certain task. %Put some refs on practicla learning w photonic circuits.

To the best of our knowledge, research in this direction has focused almost exclusively on finite-dimensional systems so far, showing that, surprisingly, one can learn many properties of a $d$-dimensional quantum state using a number of samples that scales only as $\log d$, i.e.,\emph{ linearly in the system size}. Unfortunately, most of these results cannot make any prediction for the case $d=\infty$, leaving open a quite looming question: can CV quantum systems be learned efficiently?

In particular, Aaronson~\cite{Aaronson2007} was the first to prove that it is possible to \emph{learn efficiently} an approximation of an unknown quantum state by estimating the expected values of binary observables, with a SC linear in the number of qubits; instead, the SC turns out to be exponential in the number of qubits for the dual problem of learning an unknown quantum measurement~\cite{Cheng2015}, and a similar scaling is believed to hold for learning general quantum processes~\cite{Aaronson2007,Caro2020a}. By restricting to the class of polynomial-depth quantum circuits, Caro and Datta~\cite{Caro2020a} and Popescu~\cite{Popescu2021} showed that efficient process learnability can be recovered. More recently, Caro et al.~\cite{Caro2021b} computed the SC of circuits when employed as encoders of classical information into the quantum device. 

The only known estimates for CV systems were introduced by~\cite{Arunachalam2021}, which provide energy-dependent SC bounds for \emph{online learning} of Gaussian states, a task which has in general larger SC than the ones treated here.  {\color{blue} Furthermore, after this article was published on a preprint server, three more works appeared providing methods to perform shadow tomography of CV states~\cite{Gandhari2022,Becker2022}. This is a task related with the one we consider here, though more difficult, in the sense that it requires the estimation of $M$ observables with high precision, whereas here we want to obtain estimates which are precise with high probability, on average with respect to a distribution over observables (see~\cite{Aaronson2018} for a detailed description of the differences between these two settings).
Hence one can expect that the sample complexity obtained in these further works is larger than the one given in the present manuscript. Observe also that each of these works quantifies non-Gaussianity in a different way, suitable in each case to the specific methods employed to bound the sample complexity.
A similar reasoning applies to the CV learning no-free-lunch theorem in~\cite{Volkoff2021b}, where the authors  optimize case-dependent cost functions that can be computed by applying a circuit to the target unknown state or gate. In our case, instead, one is able to find an approximation of the measurement statistics obtained by applying random circuits to the target unknown state or gate, which is a more expensive task in principle. }

In this article, we introduce a statistical learning theory for infinite-dimensional quantum systems, studying the learnability of CV quantum states and processes, e.g., those that can be implemented in a photonic processor. We show that the simple class of Gaussian states and operations can be learned efficiently, with a SC scaling quadratically in the number of bosonic modes $n$ (independent degrees of freedom, analogous to the number of qubits) of the system, i.e., also linearly in the system size and without any dependence on cutoffs. Unfortunately, it is well-known that non-Gaussian operations are required to realize a universal CV quantum computer~\cite{Lloyd1999}. Therefore, we consider next a wider class of operations including Gaussian, photocounting and generalized Gaussian (GG) operations, the latter being a non-Gaussian set introduced by Bourassa et al.~\cite{Bourassa2021}, which can approximate several key non-Gaussian resources, e.g., Gottesman-Kitaev-Preskill (GKP), cat and Fock states.
We show that also this class is learnable with a SC quadratic in the number of modes, hence quadratically in the system size, provided that its parameters satisfy certain constraints, scaling with the systems' energy, thus establishing efficient theoretical learnability for a wide range of non-Gaussian CV quantum devices. Interestingly, at variance with finite-dimensional systems, in the CV case it appears that the SC does not directly depend on the circuit's depth, but rather on the degree of non-Gaussianity of the circuit. 

The article is structured as follows: in Sec.~\ref{sec:probSetting} we describe in detail our problem setting, defining two different learning tasks for which our results apply, namely the standard learning an approximation of an unknown state or process, and the more general learning of a near-optimal circuit for tasks with linear loss function in the outcome probabilities, e.g., state discrimination or synthesis (\ref{sec:learningSetting}). We then do a brief recap of CV circuits, including Gaussian components, photodetection and a recently introduced class of non-Gaussian components (\ref{sec:cv}). Finally, we provide a compact presentation of our main results, stating SC bounds for both learning tasks with several possible CV circuit classes that arise from combining states, measurements and channels of different kinds (\ref{sec:summary}). In Sec.~\ref{sec:methods} we present our methods in detail and prove our results: we start with a statement of key theorems in classical statistical learning of probability-valued functions (\ref{sec:classStatLearn}), then proceed by computing complexity measures for three CV circuit classes (\ref{sec:compBounds}) and finally apply these results to the computation of SC bounds for both tasks previously introduced (\ref{sec:learnProofs}, \ref{sec:predProofs}). In Sec.~\ref{sec:conclusions} we discuss the significance of our results and elaborate on future research directions.

\section{Problem setting and results summary}\label{sec:probSetting}

\subsection{Learning quantum circuits from measurement outcomes}\label{sec:learningSetting}
\subsubsection{Learning quantum states, measurements and channels}\label{sec:stateLearn}
Consider a quantum circuit comprising preparation, evolution and measurement. The circuit is described by an input quantum state $\rho$, with $\tr\rho=1$ and $\rho\geq0$, a noisy quantum channel $\Phi$, i.e., a completely positive and trace-preserving map, and a measurement operator $\1\geq M\geq0$~\footnote{Note that, as typical in the learning setting that we consider, one is guaranteed to find an algorithm that approximates the expected value of $M$, i.e., the statistics of a binary-outcome measurement $\{M,\1-M\}$, \emph{on average} over $M$ sampled in a certain sample class $\cS$. Similarly, the statistics of multi-outcome measurements $\{M(m)\}_{m\in\cM}$ can be approximated using the same algorithm and by including all measurements operators $M(m)$ in the sample class $\cS$. If one however is interested in obtaining good approximations of all $P(M(m)|\Phi,\rho)$ at once, then the SC will scale at least linearly in $|\cM|$, as pointed out in~\cite{Aaronson2007} via a union-bound argument. An improvement over this scaling is conceivable only by allowing the algorithm to perform joint measurements directly on the quantum samples or by changing the requirements, e.g., as in shadow tomography~\cite{Aaronson2018}.}. When running the circuit with a given setup $(\rho,\Phi,M)$, the outcome is a bit $b\in\{0,1\}$, {\color{blue} taking the value $b=1$ when $M$ is ``accepted'', which happens with probability $P( M|\Phi,\rho)=\tr{ M\cdot\Phi( \rho)}$, as prescribed by the Born rule. }

Suppose now that $ \rho$ is unknown and we want to approximate it with a state $ \sigma\in\cC$, where $\cC$ is a class of hypotheses. We measure the quality of our approximation by how well the hypothesis $\sigma$ is able to reproduce the statistics obtained from $ \rho$ by applying channels and measurements from a class $\cS$, sampled according to a probability distribution $Q(\Phi, M)$, via the true loss function
\begin{equation}
    L_{\sigma,\rho}(\Phi,M) := \abs{P( M|\Phi, \sigma)-P( M|\Phi, \rho)}.
\end{equation}
{\color{red} Naturally, the quality of the hypotheses will depend on the employed hypothesis class: in statistical learning theory, one typically distinguishes between the \emph{realizable} setting, where the unknown state $\rho\in\cC$, and the more general \emph{agnostic} setting, where $\rho\notin\cC$.  Our results are applicable to both these learning settings; for example, we show that it is possible to find a Gaussian approximation to a non-Gaussian state efficiently. Other constraints on the hypothesis class arise naturally from the theory of CV systems, see below Sec.~\ref{sec:cv} and \ref{sec:summary}.}

\begin{figure}[ht!]
\centering
\includegraphics[width=.8\textwidth,trim={0 2cm 0 4cm},clip]{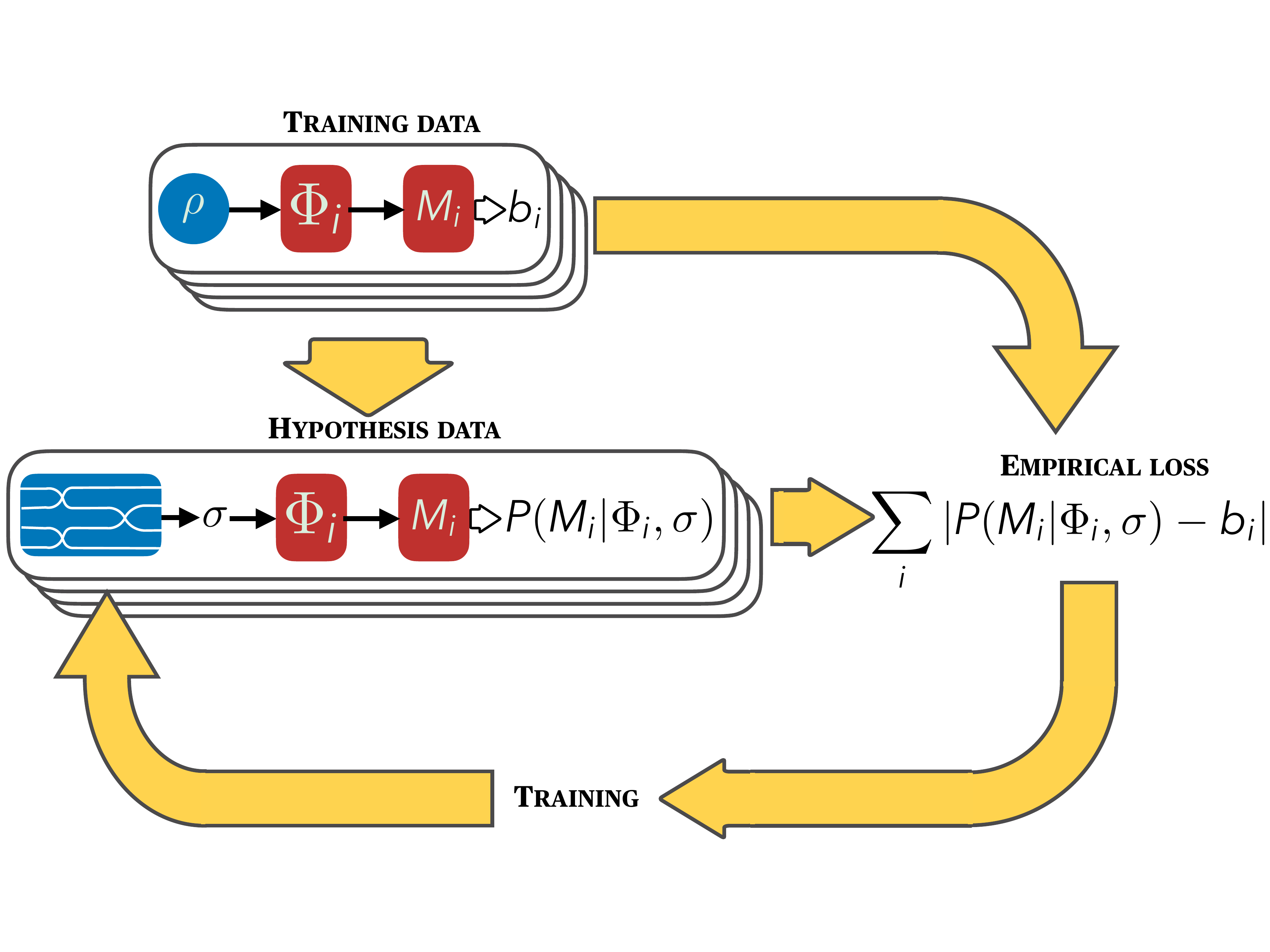}
\caption{Depiction of the simplest learning problem studied for CV architectures. We aim to reproduce the statistics of an unknown state $\rho$ {\color{blue} with respect to random couples of CV channels and measurements $(\Phi_i, M_i)$. For each $i$, the state is passed through channel $\Phi_i$ and measured with a binary measurement $M_i$ obtaining outcome $b_{i}$.} Then the learner produces a hypothesis $\sigma$ by tuning the parameters of a CV circuit, e.g., a photonic chip, and obtains its output statistics on the  channels and measurements of the training set. A suitable empirical loss evaluates how well $\sigma$ approximates $\rho$ on the training set. This can be used to optimize $\sigma$ iteratively.}\label{fig:1}
\end{figure}

In practice, as shown in Fig.~\ref{fig:1}, we will rely on a finite training set of binary outcomes $b_{i}\in\{0,1\}$ obtained by running the circuit several times, each with a different sample channel $\Phi_{i}$ and measurement operator $ M_{i}$, for $i=1,\cdots,T$. {\color{blue} Therefore, the best we can do is trying to find a candidate $\sigma\in\cC$ with small loss on each training instance $(\Phi_i,M_i,b_i)$, as measured by the empirical loss function
\begin{equation}
    \hat L_{\sigma,\rho}(\Phi_i,M_i,b_i) := \abs{P( M_i|\Phi_i, \sigma)-b_i}.
\end{equation}
A learning theorem then guarantees that the chosen hypothesis $\sigma$ can approximate the statistics of the unknown $\rho$ in terms of the true loss function, even on unseen instances, a property called generalization. 
The minimum number of samples $T_0$ that guarantees generalization will depend on the specific precision and error thresholds requested and on the chosen hypothesis class. We thus arrive at the following operational learnability statement:
\begin{definition}\label{def:stateLearn} (State learnability)
Let $\rho$ be an unknown quantum state, $\cC$ a set of hypothesis quantum states (not necessarily containing $\rho$) and $\cS$ a set of channel-measurement couples. Suppose there is an algorithm that takes as input $T\geq T_{0}$ training samples $\{(\Phi_{t}, M_{t},b_{t})\}_{t=1}^{T}$ according to the unknown distribution $Q(\Phi, M)\cdot P( M|\Phi, \rho)$, and outputs a hypothesis $ \sigma\in\cC$ that is $\eta$-good on each training instance, i.e.,
\begin{equation}
    \hat L_{\sigma,\rho}(\Phi_i,M_i,b_i) < \eta\; \forall i=1,\cdots,T.
\end{equation}
We say that an approximation to $\rho$ in $\cC$ with respect to sample channels and measurements in $\cS$ can be learned if,  with probability at least $1-\delta$ with respect to the training sample, i.e., $\{(\Phi_{t}, M_{t},b_{t})\}_{t=1}^{T}\sim Q^T$, it holds 
\begin{equation}\label{eq:learnabilityStates}
    \Pr_{(\Phi,M)\sim Q}(L_{\sigma, \rho}(\Phi,M) < \eta + \gamma)\geq 1-\epsilon,
\end{equation}
for all $\eta, \gamma,\epsilon,\delta>0$ and distributions $Q$ on $\cS$. 
The sample complexity (SC) is the minimum number of samples $T_{0}(\delta,\epsilon,\gamma, \cC,\cS)$ sufficient for learnability \eqref{eq:learnabilityStates}.
\end{definition}
This statement is a generalization of~\cite{Aaronson2007} to the agnostic case, where $\rho\in\cC$ and $\eta$ cannot be decreased arbitrarily. In the simpler realizable case, one can bring $\eta$ as close to zero as desired, thus converging to the true quantum state $\rho\in\cC$.}

Note that the only role of the quantum circuit in this setting is to provide binary outcomes $b_{i}$ according to the Born rule. Hence, one can easily interchange the role of states, channels and measurements as hypotheses and samples, obtaining similar definitions of
\begin{enumerate}[(i)]
\item \emph{Measurement learnability}~\cite{Cheng2015}, where the objective is to approximate an unknown measurement, $\cC=\{ M\}$ is a class of measurement operators and $\cS=\{( \rho,\Phi)\}$ is a class of states and channels;
\item \emph{Channel learnability}~\cite{Aaronson2007,Caro2020a}, where the objective is to approximate an unknown channel, $\cC=\{\Phi\}$ is a class of channels and $\cS=\{( \rho, M)\}$ is a class of states and measurement operators. 
\end{enumerate}
%Furthermore, even more complex hypothesis and sample classes, both comprising states, channels and measurements, can be treated in the same way. 
In Sec.~\ref{sec:methods} we provide SC bounds for these learning problems in the case where $\cC$ and $\cS$ are classes of CV states, channels and measurements. 

\subsubsection{Learning circuits for discrimination and synthesis}\label{sec:circuitTrain}
So far we have considered the problem of finding a quantum state (or, more generally, circuit) that approximates the measurement statistics of an unknown target state (resp. circuit), based on training samples received from the latter. 
{\color{blue} On the other hand, many problems in quantum information theory have the objective of finding a family of quantum circuits that carries out a certain task with optimal performance, on average with respect to different possible realizations, for example:
\begin{enumerate}[(i)]
\item \emph{Discrimination}: there is a set of states $\{\rho_{x}\}_{x\in\cX}$ and a probability distribution $Q(x)$ over $\cX$. We receive an unknown sample $\rho_{x}$, with $x\sim Q$, and the task is to find a multi-outcome measurement $\{ M_{x}\}_{x\in \cX}$, with operators in a desired measurement class $\cC$, that guesses $x$ with maximum average success probability, i.e., 
\begin{equation}
    P_{\rm succ}(\cC):=\max_{\{M_x\}_{x\in\cX}:M_x\in\cC \forall x}\E{x\sim Q}{\tr{M_{x}\rho_{x}}}. 
\end{equation}
In particular, within the CV setting considered in this paper, the discrimination of Gaussian states with non-Gaussian resources plays a major role in quantum optical communication~\cite{Rosati2017,Fanizza2020b,Bilkis2021a}.
Here we will focus on a variational approach to the problem, where the measurement class can be equivalently described as an information-processing channel, represented by a noisy quantum circuit $\Phi$ in a given channel class $\cC$, followed by a rank-$1$ projective measurement in a fixed basis $\{\dketbra{x}\}_{x\in X}$, whose outcome $x$ corresponds to guessing for the state $\rho_x$. The optimization then yields an optimal circuit within the chosen class, with average success probability
\begin{equation}\label{eq:discrimination_as_channel_optimization}
    P_{\rm succ}(\cC):=\max_{\Phi\in\cC}\E{x\sim Q}{\bra{x}\Phi(\rho_{x})\ket{x}}. 
\end{equation}
Note that, by a Naimark-dilation argument, this setting can also include POVM measurements realizable by embedding the states in a larger Hilbert space. Furthermore, one can also consider higher-rank projections or coarse-grainings of the measurement outcomes by convex combination of the rank-$1$ measurement operators, i.e., $M_x = \sum_{y\in\cX} p(y|x) \dketbra{y}$ with fixed combination coefficients.

This approach is extremely flexible, e.g., it can accomodate also channel discrimination problems. In this case, the randomness can be attributed to the channel itself, i.e., the training set is constructed by sampling from a set of channels $\{\Phi_x\}_{x\in \cX}$ that we want to discriminate, labelled by $x\sim Q$. We can optimize the input state and output measurements in a certain class $\cC=\{(\rho,\{M_x\}_{x\in\cX})\}$, with the aim of maximizing the average success probability:
\begin{equation}
    P_{\rm succ}(\cC) = \max_{(\rho,\{M_x\})\in\cC} \E{x\sim Q}{\tr{M_{x}\Phi_x(\rho)}}.
\end{equation}
\item \emph{Pure-state synthesis and embedding problems}: starting from a noisy input state sampled from a set $\{\rho_x\}_{x\in X}$ via a classical random variable $x\sim Q$, we want to produce a corresponding pure target state $\{\dketbra{\psi_{x}}\}_{x\in \cX}$ for each $x$. We can do so via a variational noisy quantum circuit, represented by a quantum channel $\Phi$ from a desired class $\cC$. The circuit is chosen to maximize the average fidelity within the class, i.e., 
\begin{equation}\label{eq:ave_fidelity_synthesis}
    F(\cC):=\max_{\Phi\in\cC}\E{x\sim Q}{\bra{\psi_{x}}\Phi(\rho_{x})\ket{\psi_{x}}}.
\end{equation}
Note that, by properly choosing the input and target states, as well as the variational circuit class, one can adapt this setting to several quantum information processing tasks, e.g., distillation and error correction problems. For example, in a distillation problem one starts with a mixed state $\rho_x=\sigma_x^{\otimes k}$, where $\sigma_x$ is not maximally entangled of which $k$ copies are available, and wants to convert it, via an LOCC channel $\Phi$, into a pure maximally entangled target state $\ket{\psi_x}$ (similarly, one can think of distillation in other resource theories, e.g., that of coherence~\cite{Diaz2020}). The random variable $x$ represents different preparations of the mixed state and different target states. The figure of merit \eqref{eq:ave_fidelity_synthesis} quantifies the average distillation fidelity. Instead, in error correction, the states $\rho_x$ might represent the result of a noisy quantum computation, depending on a classical input $x$, and the objective is to find a channel $\Phi$ that can correct the errors, mapping the states to the ideal output of the computation $\ket \psi_x$. The average error-correction fidelity is then \eqref{eq:ave_fidelity_synthesis} too.
Finally, by setting the input state $\rho_x=\dketbra{x}$ as a computational-basis encoding of the random variable $x$, we can optimize $\Phi$ as an embedding map that tries to encodes the classical variable $x$ into a target quantum $\ket{\psi_x}$  via a complex circuit, with average encoding fidelity given again by \eqref{eq:ave_fidelity_synthesis}.
\end{enumerate}

\begin{figure}[ht!]
\centering
\includegraphics[width=.9\textwidth,trim={0 2cm 0 4cm},clip]{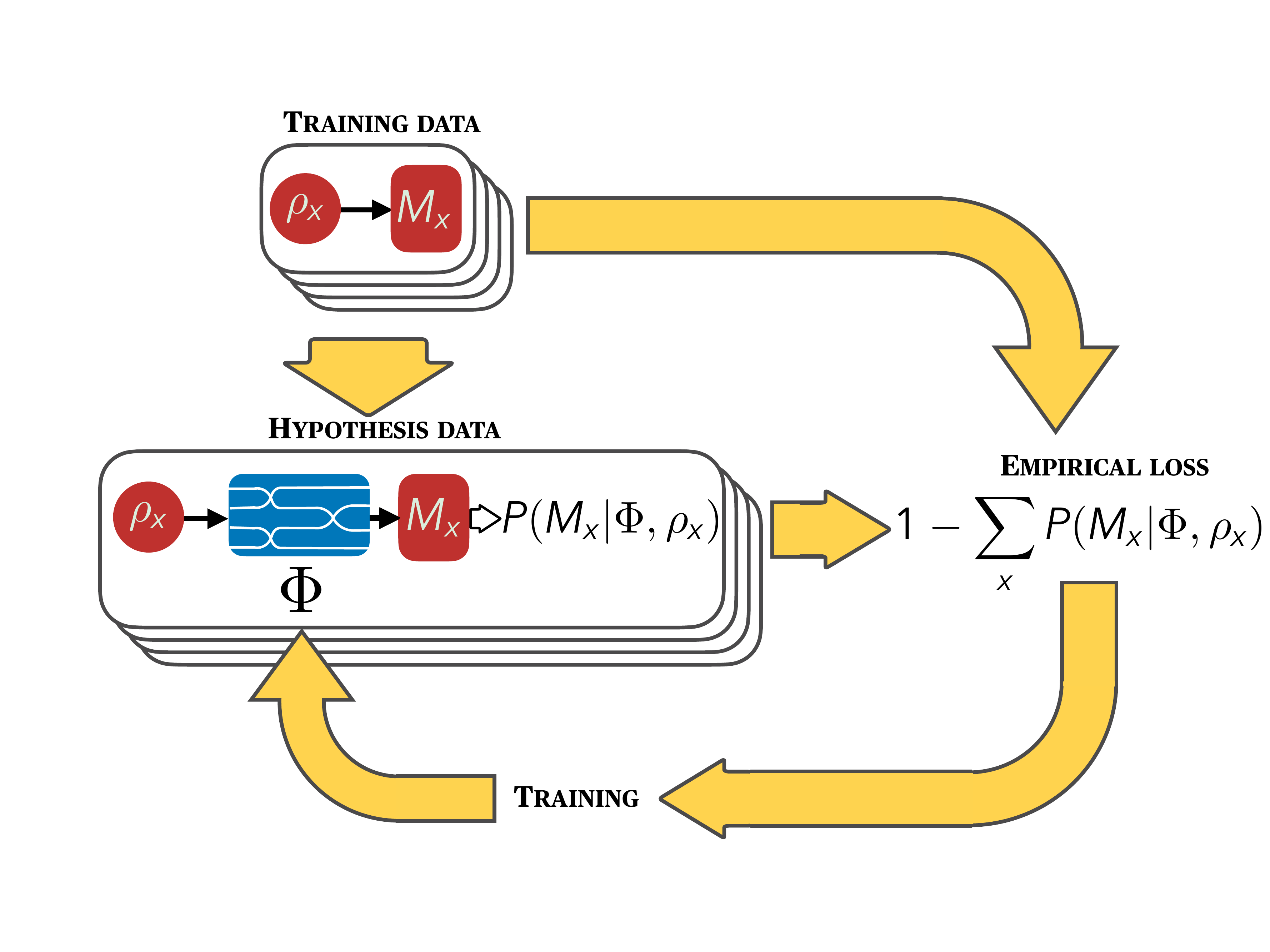}
\caption{{\color{blue}Depiction of the discrimination task learning problem studied for CV architectures. We receive random samples $(\rho_x,M_x,b_x)$ from a class of input states $\{\rho_x\}$ to be discriminated with a fixed measurement device $\{M_x\}$. Then the learner produces a hypothesis channel $\Phi$ that processes the input $\rho_x$ before measurement, and obtains its output statistics. A suitable empirical loss evaluates how well $\Phi$ is able to discriminate the states in the training set. This can be used to optimize $\Phi$ iteratively.}}\label{fig:2}
\end{figure}

%All these tasks can be described compactly as a prediction problem, also called \emph{learning a decision rule} by Kearns and Schapire~\cite{Kearns1994}.  
Traditional methods in the literature tackle these problems from a standard optimization perspective, both analytical and numerical~\cite{Nakahira2021a,Nakahira2021,Sidhu2021a,Becerra2013a,DiMario2022,Rosati17c,Assalini2011}. We show instead that it is possible to formulate a learning problem where a quantum circuit is trained to maximize the average acceptance probability based on samples $x\in\cX$, as shown in Fig.~\ref{fig:2}. The learner receives training samples of the form $(\rho_x,M_x)$, where $x\sim Q$, and tries to classify each input by assigning a bit value $b=1$ with probability $P(M_x|\Phi,\rho_x)$, where $\Phi$ represents a noisy quantum circuit in a given class $\cC$, or $b=0$ otherwise. The target decision rule is trivial, in the sense that each couple $(\rho_x,M_x)$ is always accepted, and the loss is measured by 
\begin{equation}
    L_{\Phi}(\rho,M) := (1-P(M|\Phi,\rho)).
\end{equation}
A learning guarantee then provides conditions for the identification of a hypothesis $\Phi$ whose average loss is close to the minimum in the class, given by
\begin{equation}
    \min_{\Phi\in\cC}\E{x\sim Q}{L_\Phi(\rho_x,M_x)}=1-\max_{\Phi\in\cC}\E{x\sim Q}{P(M_x|\Phi,\rho_x)},
\end{equation}
which is equivalent to finding a hypothesis with close-to-optimal success probability or fidelity in the above-described discrimination and synthesis problems. Clearly, in most applications the learning problem will be intrinsically agnostic, since an ideal circuit that classifies each state correctly, i.e., such that $P(M_x|\Phi,\rho_x)=1$ for all $x\in \cX$ might not exist at all. This brings us to the following operational task learnability statement:
% Therefore, the dependence of the states and measurements on $x$ will be in general encoded by a vector of functions $\bgamma(x)$, which can increase the complexity of the sample class $\cS=\{(\rho_{\bgamma_{1}(x)},M_{\bgamma_{2}(x)})\}$.

\begin{definition}\label{def:circuitTrain} (Optimal task learnability)
Let $\cS=\{(\rho_x,M_x)\}_{x\in \cX}$ be a set of (not necessarily known) state-measurement couples associated with a quantum-information-processing task, $Q$ a distribution on $\cX$ and $\cC$ a set of  hypothesis quantum channels. Suppose there is an algorithm that takes as input a training sample $\{(\rho_{x_t}, M_{x_t})\}_{t=1}^{T}$ of size $T\geq T_{0}$ according to the unknown distribution $Q(x)$, and outputs a hypothesis $ \Phi\in\cC$ that minimizes the cumulative empirical loss on the sample, i.e.,
\begin{equation}
    \sum_{t=1}^T L_{\Phi}(\rho_{x_t},M_{x_t}).
\end{equation}
We say that a near-optimal circuit in $\cC$ for the task represented by  $\{(\rho_x,M_x)\}_{x\in X}$ can be learned  if,  with probability at least $1-\delta$ with respect to the training sample, i.e., $\{(\rho_{x_t}, M_{x_t})\}_{t=1}^{T}\sim Q^T$, it holds 
\begin{equation}\label{eq:learnabilityTask}
    \E{x\sim Q}{L_\Phi(\rho_x,M_x)} - \inf_{\Gamma\in\cC}\E{x\sim Q}{ L_{\Gamma}(\rho_x,M_x)} \leq \epsilon
\end{equation}
for all $\epsilon,\delta>0$ and distributions $Q$ on $\cX$. 
The sample complexity (SC) is the minimum number of samples $T_{0}(\delta,\epsilon, \cC,\cS)$ sufficient for learnability \eqref{eq:learnabilityTask}.
\end{definition}
Thanks to the similarity of the induced hypothesis classes employed in Definitions~\ref{def:stateLearn} and \ref{def:circuitTrain}, in Sec.~\ref{sec:methods} we are able to derive SC bounds also for the latter learning task, in the case of CV quantum circuits. Previous combinatorial-dimension bounds can be combined with our results to obtain SC bounds for information-processing tasks also in finite-dimensions~\cite{Aaronson2007,Cheng2015,Caro2020a}.
%Finally, note that in this case one has to consider also the complexity of the encoding of $x$ into the parameters defining the circuit $C(x)$.
}

\subsection{Continuous-variable quantum information}\label{sec:cv}
A CV bosonic quantum system can be described as a quantum harmonic oscillator with position and momentum operators $\xi_{1}, \xi_{2}$ (also called quadratures), satisfying the canonical commutation relations $[\xi_{1}, \xi_{2}]=i\hbar$, and Hamiltonian $ H := \frac{h\nu}{2} (\xi_{1}^{2} +  \xi_{2}^{2})$, where $\nu$ is the oscillator's natural frequency, $h$ is Planck's constant and $\hbar=\frac{h}{2\pi}$ (for simplicity, we henceforth set $h\nu=1=\hbar$~\footnote{This amounts to measuring energy in units of $h\nu$ and time in units of $1/(2\pi\nu)$.}).  The oscillator can model, for example, the state of a mechanical resonator or a travelling electromagnetic pulse with fixed polarization, wavevector and frequency. 
A single CV quantum system is called mode and it plays an analogous role to that of a qubit for two-dimensional quantum systems. A system of $n$ modes has a vector of quadrature operators $\bxi:=( \xi_{1},  \xi_{2}, \cdots,  \xi_{2n-1},  \xi_{2n})^{T}$, were $T$ is the transpose~\footnote{\color{red} Given the limited use of quadrature operators in the article, we use greek letters to represent operator-valued quantities and latin letters to represent complex-valued quantities.}. 

A compact way to describe CV systems is via their Wigner function, a quasi-probability distribution defined for every operator $ O$ of the Hilbert space as
\be
W_{ O}(\br) := \int_{\R^{2n}} \frac{d^{2n}\bs}{(2\pi)^{n}}\, e^{-\i\bs^{T}\Omega\br}\,\tr{ O \cdot  D(\bs)^{\dag}},
\end{equation}
where $ D(\bs) := \exp(\i\bs^{T}\Omega\bxi)$ is the displacement operator and $\Omega:=\bigoplus_{i=1}^{n}\left(\begin{array}{cc}0&1\\-1&0\end{array}\right)$ is the symplectic form. Note that, if $O$ is Hermitian, its Wigner function is real-valued; if additionally $ O = {\rho}$ is a normalized quantum state, then its Wigner function is normalized as well, i.e., $\int_{\R^{2n}} d^{2n}\br W_{\rho}(\br) = 1$. 

As observed in Sec.~\ref{sec:learningSetting}, our central object of interest will be the outcome probability of a CV quantum circuit, which can be easily determined via integration of the Wigner functions of the final circuit state $\rho$ and the measurement operator $M$:
\be
P(M|\rho)=\tr{M\rho}=\int_{\R^{2n}} d^{2n}\br W_{M}(\br) W_{\rho}(\br).
\end{equation}

Finally, key quantities in the analysis of CV systems are the first and second moments of the quadratures, i.e., the mean-values vector $\bfm \in\R^{2n}$ and covariance matrix $V\in{\rm Sym}_{2n}(\R)$\footnote{In the article we use the following notation for sets of square matrices of order $2n$ with real coefficients (and similarly for complex coefficients): ${\rm Mat}_{2n}(\R)$ is the set of all such matrices, ${\rm Sym}_{2n}(\R)$ is the subset of symmetric matrices, ${\rm Sp}_{2n}(\R)$ is the subset of symplectic matrices.} with components
\be
m_{i}:=\tr{ \xi_{i} \rho},\quad V_{i,j}:=\frac12\tr{( \xi_{i}  \xi_{j} +  \xi_{j}  \xi_{i})\rho}.
\end{equation}
Note that the covariance matrix of a quantum state must respect the uncertainty principle $V+\i\Omega\geq0$, implying $V\geq0$.

\subsubsection{Gaussian circuits}
The simplest class of CV circuits comprises only Gaussian states, channels and measurements. Indeed, any element of this class can be experimentally realized using lasers, optical instrumentation and photodetectors; furthermore, a programmable integrated-photonic chip including all such elements is well within reach of current technology~\cite{Madsen2022,Wang2020,EliBourassa2021}. 
Gaussian states $\rho_{\bfm,V}$ on $n$ modes are completely determined by their mean $\bfm\in\R^{2n}$ and covariance matrix $V\in{\rm Sp}_{2n}(\R)$: their Wigner function is Gaussian with the same mean and covariance matrix $V$, i.e.,
\be
W_{ \rho_{\bfm,V}}(\br)=\cG_{\bfm,V}(\br):=\frac{e^{-\frac12(\br-\bfm)^{T}V^{-1}(\br-\bfm)}}{\sqrt{\det{2\pi V}}}.
\end{equation}
Similarly, general-dyne Gaussian measurement operators $ M_{\bfm,V}$ have Wigner function $W_{ M_{\bfm,V}}(\br)=\cG_{\bfm,V}(\br)$, the most common ones being homo- and hetero-dyne measurements (the latter corresponding to $V=\frac{\1}{2}$ and arbitrary $\bfm$). %One can construct a full Gaussian measurement with these operators by fixing the covariance matrix $V$ and labelling the outcomes via the mean $\bfm$, i.e., $\{ M_{\bfm,V}\}_{\bfm\in\R^{2n}}$. 
Finally, a noisy Gaussian channel is determined by a displacement vector $\bd\in\R^{2n}$ and two matrices $X\in{\rm Mat}_{2n}(\R)$, $Y\in{\rm Sym}_{2n}(\R)$ satisfying the constraint $Y+\i\Omega\geq X\Omega X^{T}$, which amounts to preserving the uncertainty principle. The action of a Gaussian channel $\Phi_{\bd,X,Y}$ preserves Gaussian states, transforming their mean and covariance as follows:
\begin{equation}\label{eq:output_Gaussian_channel}
\bfm \mapsto \bfm_{\rm out}:=X\bfm +\bd,\quad V \mapsto V_{\rm out}:=XVX^{T}+Y.
\end{equation}

Therefore, the outcome probability of a Gaussian circuit is
\begin{equation}\label{eq:g}
P_{\rm g}(M_{\bfm',V'}|\Phi_{\bd,X,Y},\rho_{\bfm,V})=\cG_{\bfm_{\rm out},V_{\rm out}+V'}(\bfm').
\end{equation}

\subsubsection{Photodetection}\label{subsubsec:gpDef}
Non-Gaussian resources are needed to realize universal quantum computation with CV systems. One such resource is photodetection, whose single-mode measurement operators $M_{k}:=\dketbra{k}$ correspond to projections onto Fock states $\ket k$, i.e., states with a fixed number $k$ of photons (or excitations) such that $ H\ket k = k +\frac12$. If we allow the detection of up to $K$ photons in a single mode, an $n$-mode photodetection measurement is hence described by a vector of outcomes $\bk\in\{0,\cdots,K\}^{n}$ as $M_{\bk}:=\dketbra{\bk}=\bigotimes_{i=1}^{n}\dketbra{k_{i}}$. 

The outcome probability of a photodetection measurement on a Gaussian circuit is
\begin{equation}\label{eq:f}
P_{\rm gp}(M_{\bk}|\Phi_{\bd,X,Y},\rho_{\bfm,V})=(2\pi)^{n}\cG_{\bfm_{\rm out},V_{\rm out}+\frac\12}(0) H_{\bk}^{(\tilde V)}(\tilde\bfm),
\end{equation}
with $H_{\bk}$ a multivariate Hermite polynomial~\cite{Rahman2017} defined as 
\be
H_{\bk}^{(V)}(\bfm):=\cG_{0,V}(\bfm)^{-1}\prod_{i=1}^{n}\left(-\frac{\partial}{\partial \bm_{2i-1}}\right)^{\bk_{i}}\left(-\frac{\partial}{\partial \bm_{2i}}\right)^{\bk_{i}}\cG_{0,V}(\bfm),
    \end{equation}
while $\tilde V\in{\rm Mat}_{2n}(\R)$ and $\tilde\bfm\in\R^{2n}$ are quadratic functions of $V_{\rm out}$ and $\bfm_{\rm out}$~\cite{Dodonov1994}. In the following, we refer to Gaussian plus photodetection (GP) circuits as those obtained by combining Gaussian channels with photodetection measurements, possibly coarse-grained. The latter term refers to a post-processing of outcomes that lumps together different photocounting events $\bk$ with probabilities $q_{\bk}\geq0$ such that $\sum_{\bk}q_{\bk}\leq1$, resulting in a measurement operator $M_{q}:=\sum_{\bk}q_{\bk}\dketbra{\bk}$ and clearly
\begin{equation}
P_{\rm gp}(M_{q}|\Phi_{\bd,X,Y},\rho_{\bfm,V})=\sum_{\bk}q_{\bk}P_{\rm gp}(M_{\bk}|\Phi_{\bd,X,Y},\rho_{\bfm,V}).
\end{equation}
%$R=U^{\dag}(\1- V')(\1+ V')^{-1}U^{*}$, $\br=2U^{T}(\1- V')^{-1}\bx'$ and
%\be
%U=\frac{1}{\sqrt{2}} \left(\begin{array}{cc}
%-i\1 & i \1 \\
%\1 & \1
%\end{array}\right).
%\end{equation}

\subsubsection{Generalized Gaussian circuits}
More generally, Bourassa et al. \cite{Bourassa2021} introduced for simulation purposes a class of non-Gaussian circuits that we call generalized Gaussian (GG). An arbitrary GG state is defined by a possibly infinite set of complex-valued coefficients, means and covariances $\{c_{i}\in\C,\bfm_{i}\in\C^{2n},V_{i}\in{\rm Mat}_{2n}(\C)\}_{i\in\cI}$. The Wigner function of a GG state $ \rho_{\{c,\bfm,V\}_{\cI}}$ is then given by a complex combination of complex-valued Gaussians:
\begin{equation}\label{eq:GG}
W_{ \rho_{\{c,\bfm,V\}_{\cI}}}(\br):=\sum_{i\in\cI} c_{i} \cG_{\bfm_{i},V_{i}}(\br),% + {\rm c.c.},
\end{equation}
where we omitted the explicit dependence on $i$ of all the parameters on the left-hand side, for simplicity.
%where, without loss of generality, we have explicitly added the complex conjugate of each term, since the Wigner function of a hermitian operator is real.
Notice that, while the total Wigner function must be real, since $\rho$ is Hermitian, its components instead are in general complex-valued and do not correspond to physically valid quantum states. Nevertheless we require that ${\rm Re} V_{i}\geq0$ for all $i\in\cI$, so that all Gaussian components are bounded, and $\sum_{i\in\cI}V_{i}+\i\Omega\geq0$, so that the full covariance matrix respects the uncertainty principle. Furthermore, the overall Wigner function must be normalized, implying $\sum_{i\in\cI}c_{i}=1$. %+{\rm c.c.}=1$.
{\color{blue} Furthermore, we stress that this class is far from a simple mathematical abstraction. For example, it includes typical non-Gaussian states, e.g., cat states, or arbitrary-precision approximations theoreof, e.g., Fock or GKP states, which can all be written as complex superpositions of pure Gaussian states. We refer the reader to~\cite{Bourassa2021} for a detailed discussion.

Similarly, GG measurement operators $M_{\{c,\bfm,V\}_{\cI}}$ have a Wigner function like \eqref{eq:GG}, with the only difference that in general it is sub-normalized, i.e., $\sum_{i\in\cI}c_{i}\leq1$.%+{\rm c.c.}\leq 1$. %and we do not expect the full covariance matrix to satisfy the uncertainty principle.
As for the case of GG states, this class of measurements includes projections on the most relevant non-Gaussian states or approximations thereof.} 

Finally, we can perform a GG channel on an arbitrary input state $\rho$ by introducing ancillary modes in a non-Gaussian state, performing a Gaussian channel on the input plus ancillary modes, and finally performing a non-Gaussian measurement on the ancillary modes, possibly followed by another Gaussian channel conditioned on the measurement outcome. The overall action of a GG channel on a Gaussian state is determined by a set of complex-valued coefficients, displacement vectors and matrices $\{c_{i}\in\C,\bd_{i}\in\C^{2n},X_{i},Y_{i}\in{\rm Mat}_{2n}(\C)\}_{i\in\cI}$ as follows:
\be
\Phi_{\{c,\bd,X,Y\}_{\cI}}: \rho_{\bfm,V} \mapsto \rho_{\{c,\bfm_{\rm out},V_{\rm out}\}_{\cI}},
\end{equation}
with $\sum_{i}c_{i}=1$ in order to guarantee the trace-preserving property.

Therefore, the outcome probability of a GG circuit is
{\begin{align}\label{eq:gg}
&P_{\rm gg}(M_{\{c',\bfm',V'\}_{\cI'}}|\Phi_{\{c'',\bd,X,Y\}_{\cI''}},\rho_{\{c,\bfm,V\}_{\cI}})\\
&=\sum_{i\in\cI,j\in\cI',k\in\cI''} c_{i} c_{j}' c_{k}''\, \cG_{\bfm_{\rm out}^{(ijk)},V_{\rm out}^{(ik)}+V_{j}'}(\bfm_{j}'),\nonumber%+ \text{c.c.},\nonumber
\end{align}}
where $\bfm_{\rm out}^{(ijk)}:=X_{k}\bfm_{i}+\bd_{k}$, $V_{\rm out}^{(ik)}:=X_{k}V_{i}X_{k}^{T}+Y_{k}$, and the sum runs over all indices defining the GG state, channel and measurement, which will be left implicit in the rest of the article.

\subsection{Results summary}\label{sec:summary}
We are now ready to state our main result: we establish that quantum states generated by $n$-mode CV circuits can be learned efficiently in $n$, for all the practically relevant circuit classes that can be constructed with components from Sec.~\ref{sec:cv}:
\begin{theorem}\label{thm:statelearnedot}
Let $\rho$ be an unknown quantum state, $\cC_{\rm g, gg}$ the hypothesis classes of Gaussian and GG states, and $\cS_{\rm g, gp, gg}$ the hypothesis classes of Gaussian, GP and GG measurements and channels on an $n$-mode CV architecture.
%Let $\rho$ be an unknown CV quantum state, $\cC$ a class of CV states and $\cS$ a class of CV channels and measurements on an $n$-mode CV architecture, among the following ones:
%\begin{itemize}
%\item $\cC_{\rm g}$ ($\cS_{\rm g}$) the hypothesis (sample) class of Gaussian states (channels and measurements); 
%\item $\cC_{\rm gg}$ ($\cS_{\rm gg}$) the hypothesis (sample) class of GG states (channels and measurements);
%\item $\cS_{\rm gp}$ the sample class of Gaussian channels and photodetection measurements.
%\end{itemize}
Then it is possible to learn an $\epsilon$-approximation of $\rho$ inside $\cC$ using samples from $\cS$ (in the sense of Def.~\ref{def:stateLearn}) with SC polynomial in $n$ and $\epsilon^{-1}$. 
Explicit SC upper bounds for each of the above classes are given in Table~\ref{tab:states}, hiding the dependence on $\nu, \delta$ and up to sub-leading $\log$ factors.%, up to logarithmic factors, 
\end{theorem}
\begin{table}[!ht]
\centering
\begin{tabular}{|c|c|c|c|}
\cline{3-4}
 \multicolumn{2}{c|}{} & \multicolumn{2}{|c|}{hypothesis states} \\
\cline{3-4}
 \multicolumn{2}{c|}{} & $\cC_{\rm g}$ & $\cC_{\rm gg}$ \\
\hline
\multirow{3}{6em}{\centering Sample channels and measurements} & $\cS_{\rm g}$ & $O(n^{2} \epsilon^{-2}\log n)$ & $O(n^{2} \epsilon^{-4}B^{2})$\\
\cline{2-4}
 & $\cS_{\rm gp}$ & $O(n^{2} \epsilon^{-2} \log(n K))$ & unknown \\
\cline{2-4}
 & $\cS_{\rm gg}$ & $O(n^{2} \epsilon^{-4}B^{2})$ & $O(n^{2} \epsilon^{-4}B^{2})$\\
\hline
\end{tabular}\caption{Sample complexity of learning Gaussian or GG states with Gaussian, GP or GG channels and measurements. Here $n$ is the number of modes of the CV circuit, $\epsilon$ the error, $K$ the photon-number cutoff and $B$ a constant restricting the $GG$ class.}\label{tab:states}
\end{table}

Here $K$ is the maximum photon-number cutoff for photodetection measurements, while $B$ is a constant restricting the GG class (see~\ref{sec:ggCovNum}). 
Note that the same SC bounds of Table~\ref{tab:states} hold for channel and measurement learning, i.e., swapping the roles of $\cS$ and $\cC$, except for one notable difference: photodetection measurement learning with Gaussian states and channels has SC scaling exponentially with the number of modes, $O(K^{n})$ (see \ref{sec:gpPdim}).

Similarly, we establish learnability of CV circuits for discrimination and synthesis tasks with polynomial encoding functions:
\begin{theorem}\label{thm:circuitTrainTot}
Let $\cC$ be a class of $n$-mode CV quantum channels and $(\rho_{\bgamma(x)},M_{\bgamma(x)})_{x\in\cX}$ the states and measurements describing a discrimination or synthesis task, where $\bgamma(x)$ is a vector of encoding functions that are polynomial in $x$ of order at most $\ell$. In particular, let $\cC_{\rm g,\rm gp, \rm gg}$ be the classes of Gaussian, GP and GG circuit families.
%Suppose furthermore that the encoding functions are bounded, i.e., $|\bgamma_{i}(x)|\leq B_{0}$ for all components $i$ and $\bgamma(x)\in\cE_{\ell}$.

Then it is possible to learn an $\epsilon$-optimal channel in $\cC$ for the task described by $(\rho_{\bgamma(x)},M_{\bgamma(x)})_{x\in\cX}$ (in the sense of Def.~\ref{def:circuitTrain}) with SC polynomial in $n$, $\epsilon^{-1}$, $\ell$. Explicit SC upper bounds are given in Table~\ref{tab:tasks}.%, up to logarithmic factors, 
\end{theorem}
\begin{table}[!ht]
\centering
\begin{tabular}{|c|c|c|}
\cline{3-3}
 \multicolumn{2}{c|}{} & \multicolumn{1}{|c|}{Sample complexity} \\
\hline
\multirow{3}{3em}{\centering Circuit class} & $\cC_{\rm g}$ & $O(\ell n^{2}\epsilon^{-2}\log n)$\\
\cline{2-3}
 &$\cC_{\rm gp}$ & $O(\ell n^{2}\epsilon^{-2}\log(n K))$ \\
\cline{2-3}
 &$\cC_{\rm gg}$ & $O(\ell n^{2}\epsilon^{-4}B^{2})$\\
\hline
\end{tabular}\caption{Sample complexity of learning Gaussian, Gaussian plus photodetection, or GG circuits for discrimination and synthesis tasks. Here $n$ is the number of modes of the CV circuit, $\epsilon$ the error, $K$ the photon-number cutoff, $B$ a constant restricting the $GG$ class and $\ell$ the order of the encoding polynomials.}\label{tab:tasks}
\end{table}

\section{Sample complexity of CV circuits}\label{sec:methods}
\subsection{Bounding the SC of probability-valued functions}\label{sec:classStatLearn}
Our results build on the study of learnability for probability-valued function classes, pioneered in the works \cite{Kearns1994,Alon1997,Bartlett1998} and more recently employed in the quantum setting by \cite{Aaronson2007}. {\color{blue} The following theorems bound the number of samples required to approximate a probabilistic relation between two random variables, in terms of a suitable effective dimension of the function class employed for approximation. Thus, by computing such dimension for the various classes induced by processing and measuring CV circuits, as done in the next sections, we can bound their sample complexity.}

We start by introducing two quantities that measure the effective dimension of a function class:
{\color{blue}
\begin{definition} (Fat-shattering- and pseudo-dimension)
Let $\cF:=\{f:x\in\cX\mapsto f(x)\in\R\}$ be a class of real-valued functions on a set $\cX$, take $\gamma>0$ and a sequence of $k$ inputs $\bx\in\cX^{k}$. We say that $\bx$ is $\gamma$-shattered by $\cF$ if there exists a sequence of thresholds $\balpha\in\R^{k}$ such that, for each sequence of $k$ bits $(b_1,\cdots,b_k)\in\{0,1\}^k$ determining a binary classification pattern of the inputs $\bx$ there exists a concept $f\in\cF$ that reproduces such pattern with thresholds $\balpha$ and safety margin $\gamma$, i.e.,
\begin{equation}\label{eq:gamma_fat_shattering}
f(x_{i})\geq\alpha_{i}+\gamma\;\forall i:b_i=1\text{ and }f(x_{i})\leq\alpha_{i}-\gamma\;\forall i:b_i=0.
\end{equation}

The $\gamma$-fat-shattering dimension, $\fdim{\cF}{\gamma}$, is the largest $k$ such that there exists a sequence $\bx\in\cX^{k}$ that is $\gamma$-shattered by $\cF$.  If there is no such $k$, then $\fdim{\cF}{\gamma}=\infty$.

The pseudo-dimension is 
\begin{equation}
\pdim{\cF}:=\lim_{\gamma'\rightarrow0}\fdim{\cF}{\gamma'}\geq\fdim{\cF}{\gamma}
\end{equation}
 for all $\gamma>0$.
\end{definition}
In order to prove our results we rely on two theorems from the classical statistical learning theory of real-valued functions. In particular, a line of works originating from Kearns and Schapire~\cite{Kearns1994} (see also~\cite{Alon1997,Bartlett1998,Anthony2000}) focused on obtaining learning guarantees for so-called probabilistic concepts (p-concepts); these are functions $f:x\mapsto[0,1]$ that can be assimilated to the probability distribution of a random bit $b$ conditioned on a random input $x$, i.e., $b=1$ with probability $f(x)$ or $b=0$ otherwise. In this setting, Kearns and Schapire identified two distinct learning problems, given a hypothesis class $\cF$ of p-concepts: (i) \emph{learning a model of probability}, where one is interested in finding a hypothesis $h\in\cF$ that is close to the true p-concept $f$ on each input $x$; and (ii) the simpler problem of \emph{learning a decision rule}, where one is interested in finding a hypothesis $h\in\cF$ that minimizes the mis-classification error on average with respect to the bit $b$ and the input $x$.

Aaronson ~\cite{Aaronson2007} discussed both these learning frameworks, linking them with p-concepts originating from quantum states and measurements via the Born rule. Here we observe that problem (i), i.e., learning a model of probability, is a proxy for state learnability (see Definition~\ref{def:stateLearn}), and it requires an extension of Aaronson's results to the agnostic case:
\begin{theorem}(Agnostic p-concept learning) \label{thm:agnostic_pfunction_learning}
Let $\cX$ be a sample space, $Q$ a distribution on $\cX$ and $\cF\subseteq[0,1]^\cX$ a hypothesis class. Consider an unknown p-concept $f:\cX\mapsto[0,1]$ (not necessarily in $\cF$) and a training set $\{(x_t,b_t)\}_{t=1}^T$, where $x_t\sim Q$ and $b_t=1$ with probability $f(x_t)$ or $b_t=0$ otherwise. Suppose that we choose a hypothesis $h\in\cF$ with small quadratic loss on each element of the training set, i.e., $(h(x_t)-b_t)^2<\eta$ for all $t=1,\cdots,T$. Then with probability at least $1-\delta$ over the sample it holds
\begin{equation}\label{eq:agnostic_pfunction_learning_guarantee}
    \Pr_{x\sim Q}((h(x)-f(x))^2<\eta+\gamma)\geq 1-\epsilon
\end{equation}
for all $\eta,\gamma,\epsilon,\delta>0$, provided that 
\begin{equation}\label{eq:sample_complexity_agnostic_pfunction_learning}
    T\geq T_0(\epsilon,\gamma,\delta,\cF) = O\left(\frac1\epsilon\left(d\log^2\frac{d}{\gamma \epsilon} + \log\frac{1}{\delta}\right)\right),
\end{equation}
where $d = {\rm fat}_{\cF}\left(\frac{\gamma}{8}\right)$, 
\end{theorem}
The proof of this Theorem is provided in Appendix~\ref{app:proof_pfunction_learnability}. We can also easily switch to the linear loss, more commonly used to measure the distance between quantum circuit outputs, via the following corollary: 
\begin{corollary}\label{thm:genLearn}
Under the same assumptions of Theorem~\ref{thm:agnostic_pfunction_learning}, suppose that we choose a hypothesis $h\in\cF$ with small linear loss on each element of the training set, i.e., $|h(x_i)-b_i|<\eta$ for all $i=1,\cdots,T$. Then with probability at least $1-\delta$ over the sample it holds
\begin{equation}
    \Pr_{x\sim Q}(|h(x)-f(x)|<\eta+\gamma)\geq 1-\epsilon
\end{equation}
for all $\eta,\gamma,\epsilon,\delta>0$ with the same sample complexity as \eqref{eq:sample_complexity_agnostic_pfunction_learning}.
\end{corollary}
\begin{proof}
    We apply Theorem~\ref{thm:agnostic_pfunction_learning} with $\eta\mapsto\eta^2$ and $\gamma\mapsto\gamma^2$. Then the sample complexity remains of the same order of magnitude in $\log\frac1\gamma$ and the learning guarantee \eqref{eq:agnostic_pfunction_learning_guarantee} implies that
    \begin{equation}
   |h(x)-f(x)|<\sqrt{\eta^2+\gamma^2}\leq\sqrt{(\eta+\gamma)^2}=\eta+\gamma     
    \end{equation}
    holds with probability at least $1-\epsilon$ with respect to $x\sim Q$ and at least $1-\delta$ with respect to the sample. 
\end{proof}

On the other hand, problem (ii) above, i.e., learning a decision rule, is a proxy for task learnability (see Definition~\ref{def:circuitTrain}) and it can be solved by applying directly a Theorem by Aaronson:
\begin{theorem} (Prediction~\cite[Suppl.~Mat.~Thm.~8]{Aaronson2007})\label{thm:genPred}
Under the same assumptions of Theorem~\ref{thm:agnostic_pfunction_learning}, given $T$ samples $\{(x_{t},b_{t})\}_{t=1}^{T}$ such that  $x_t\sim Q$ and $b_t\sim f(x_t)$,  i.e., $b_t=1$ with probability $f(x_t)$ or $b_t=0$ otherwise, suppose there exists a learning algorithm that outputs a hypothesis $h\in\cF$ minimizing the empirical total-variation loss $\sum_{t=1}^{T}|h(x_{t})-b_{t}|$.  Define the mis-classification error of hypothesis $h$ as 
\begin{equation}
    \Delta_{h,f}(x) := \E{b\sim f(x)}{|h(x)-b|} = f(x) (1-h(x)) + (1-f(x)) h(x).
\end{equation}

Then with probability at least $1-\delta$ over the sample it holds
\begin{equation}
\mathbb{E}_{x\sim Q}\left[\Delta_{h,f}(x)\right]- \inf_{c\in\cF}\mathbb{E}_{x\sim Q}\left[\Delta_{c,f}(x)\right]\leq\epsilon,
\end{equation}
for all $\epsilon,\delta>0$, provided that the training set has size at least 
\begin{equation}
T=O\left(\frac{1}{\epsilon^{2}}\left(d\log^{2}\frac{1}{\epsilon} + \log\frac1\delta\right)\right),
\end{equation}
with $d=\fdim{\cF}{\frac{\epsilon}{10}}$.
\end{theorem}
In particular, by choosing a trivial target concept $f(x)=1$ one can recover the task optimization of Definition~\ref{def:circuitTrain} .
}

These Theorems guarantee that it is possible to identify ``good'' approximations to an unknown probabilistic process within a certain function class, using a number of training samples that scales linearly with the fat-shattering dimension of the class. In Secs.~\ref{sec:gPdim},\ref{sec:gpPdim} we compute the pseudo-dimension of function classes comprising output probabilities of GG and GP circuits. This immediately gives SC bounds using the Theorems above, since the pseudo-dimension is the largest among fat-shattering dimensions.

In general, however, the pseudo-dimension can be infinite while the SC is still finite. This is the case for the function class of outcome probabilities of GG circuits (Sec.~\ref{sec:ggCovNum}), hence we need to consider a tighter measure of complexity, known as covering numbers:
{\color{blue}
\begin{definition} (Covering number)
Let $\cY$  be a subset of a normed vector space with distance measure $d(y_1,y_2)$ for all $y_1, y_2\in\cY$. Then $\cW\subseteq\cY$ is an internal $\epsilon$-cover of $\cY$ if, for each $y\in\cY$ there exists a $w\in\cW$ that is $\epsilon$-close to it, i.e., $d(y,w)\leq\epsilon$.

Now let $\cF=\{f:\cX\rightarrow\cY\}$ be a class of  $\cY$-valued functions and consider a subset of the input set $\Xi\subseteq\cX$. The restriction of $\cF$ to the input subset $\Xi$ is defined as 
\begin{equation}
\cF|_{\Xi}:=\{F\in\cY^{\Xi}:\exists f\in\cF\text{ s.t. } F(x) = f(x)\; \forall x\in\Xi\},
\end{equation}
i.e., the set of distinct functions on $\Xi$ whose value coincides with that of a function from $\cF$  on all inputs in the subset.

We define the uniform $\epsilon$-covering number of $\cF$ at scale $k$ as the size of the largest $\epsilon$-cover of $\cF|_\Xi$, optimized over all input subsets $\Xi$ of size $k$, i.e.,
\begin{equation}
    \cN_d(\epsilon,\cF,k):=\max\{|\cW|: \exists \Xi\subseteq\cX,|\Xi|=k\text{ s.t. }\cW \text{ $\epsilon$-cover of  }\cF|_\Xi\},
\end{equation}
where $d$ is the distance measure on $\cY$ used to construct the covers.
\end{definition}
}
Furthermore, the covering number can be related with the pseudo-dimension as follows:
\begin{lemma} (\cite[Theorem 12.2]{Anthony1999})
Let $\cF$ be a function class with real, bounded output, i.e., $f\in[-B,B]$ for all $f\in\cF$, and $d$ any $p$-norm-distance. Then for all $\epsilon>0$ and $k\in\N$ it holds
\begin{equation}\label{eq:covPdim}
\log \cN_{d}(\epsilon,\cF,k)\leq \pdim{\cF}\log\frac{2eBk}{\pdim{\cF}\epsilon}.
\end{equation}
\end{lemma}

{\color{red} Finally, we state some useful Lemmas on the properties of the pseudo-dimension of composed classes.

The following result can be obtained by a simple generalization of the same result obtained for real function classes: 
\begin{lemma}(\cite[Lemma 5]{Haussler1992})\label{lemma:monot}
Let $\cF:=\{f:\cX\rightarrow\cY\}$ be a function class and $g:\cY\rightarrow\cZ$ a strictly monotonous function. Then $\pdim{g(\cF)}=\pdim{\cF}$.
\end{lemma}

The following result is stated by~\cite{Asor2014} and it can be obtained by combining the same result for Vapnik-Chervonenkis (VC) dimension with the characterization of pseudo-dimension in terms of VC dimension.
\begin{lemma}(\cite{Freund1997})\label{lemma:sumPdim}
Let $\cF_{3}:=\cF_{1}+\cF_{2}$, then $\pdim{\cF_{3}}=O(\pdim{\cF_{1}}+\pdim{\cF_{2}})$.
\end{lemma}

The following are standard results in classical statistical learning theory.
\begin{lemma}(\cite[Theorem 8.3]{Anthony1999})\label{lemma:poly}
Let $\cF:=\{f_{a}(x):\cX\mapsto\R\mid a\in\R^{d}, f_{a}(x)={\rm poly}_{\ell}(a)\forall x\in\cX\}$ be a class of functions that, for each fixed input $x\in\cX$, are polynomials of order $\ell$ in their $d$ real parameters. Then $\pdim{\cF}\leq 2d\log(12\ell)$.
\end{lemma}

\begin{lemma}(~\cite[Lemma 14.13]{Anthony1999})\label{lemma:lipschitz}
Let $\cF$ be a real-valued function class and $g:\R\rightarrow\R$ be $L-Lipschitz$. Then $\cN_{2}(\epsilon,g(\cF),k)\leq\cN_{2}(\frac\epsilon L,\cF,k)$. 
\end{lemma}

The following result follows from a general theorem in statistical learning, by restricting to proper convex combinations, i.e., with coefficients summing to less than one:
\begin{lemma}(~\cite[Theorem 14.14]{Anthony1999})\label{lemma:convHull}
Let $\cF$ be a function class satisfying three properties: (i) $\cF=-\cF$, (ii) $0\in\cF$ and (iii) $f\in[-B,B]$ for all $f\in\cF$. Then 
\begin{equation}
\log \cN_{2}(\epsilon,k,\mathfrak{C}(\cF))\leq\left\lceil \left(\frac{B}{\epsilon}\right)^{2} \right\rceil \log \cN_{2}(\epsilon,\cF,k).
\end{equation}
\end{lemma}

The following result regards the covering number of a product of classes.
\begin{lemma}\label{lemma:covProd}
Let $\cF_{i}\subseteq [-B_{i},B_{i}]^{X}$, $i=1,2$, be two bounded real-valued function classes on a space $X$ and define their product 
\be
\cF:=\cF_{1}\cdot\cF_{2}=\{f=f_{1}\cdot f_{2}:f_{i}\in\cF_{i}, i=1,2\}.
\end{equation}
Then
\begin{equation}
\cN(B_{1}\epsilon _{1}+B_{2}\epsilon_{2},\cF,m)\leq \cN\left(\epsilon_{1},\cF_{1},m\right)\cN\left(\epsilon_{2},\cF_{2},m\right).
\end{equation}
\end{lemma}
\begin{proof}
Suppose that $\cW_{i}$ are $\epsilon_{i}$-covers for the two classes and consider the class $\cW=\cW_{1}\cdot \cW_{2}$. Then for any $f=f_{1}\cdot f_{2}\in\cF$ and $i=1,2$ there exist $w_{i}\in \cW_{i}$ such that $d(w_{i},f_{i})\leq\epsilon_{i}$ by construction, where $d$ is the distance used to define the covers. By the triangle inequality it follows
{\begin{align}
d(w_{1}w_{2},f_{1}f_{2})& \leq d(w_{1}w_{2},f_{1}w_{2}) + d(f_{1}w_{2},f_{1}f_{2}) \\
&\leq B_{1} d(w_{1},f_{1}) + B_{2} d(w_{2},f_{2}) \leq  \epsilon \nonumber
\end{align}}
for $\epsilon=B_{1}\epsilon_{1}+B_{2}\epsilon_{2}$, which implies that $\cW$ is an $\epsilon$-cover of $\cF$ of size $|\cW_{1}|\cdot |\cW_{2}|$.
\end{proof}
}

\subsection{Pseudo-dimension and covering numbers of CV circuits}\label{sec:compBounds}
{\color{blue} In this Section we bound the pseudo-dimension or covering number of function classes representing probability distributions that arise from manipulating and measuring CV classes. These bounds, together with the learning theorems for classical probability distributions of Sec.~\ref{sec:classStatLearn}, immediately provide the sample complexity bounds given in Tables ~\ref{tab:states}, \ref{tab:tasks}.}

\subsubsection{Gaussian circuits}\label{sec:gPdim}

Consider an unknown $n$-mode CV state $\rho$ and suppose that we can probe it by applying channels and measurement operators in the Gaussian class 
\begin{equation}
\cS_{\rm g}(n):=\{(\Phi_{\bd,X,Y}, M_{\bfm',V'})\},
\end{equation} 
and obtaining measurement outcomes. We want to ``imitate'' these outcomes by finding an approximation of $\rho$ in the class of Gaussian states 
\begin{equation}
\cC_{\rm g}(n):=\{\sigma_{\bfm,V}\}.
\end{equation} 
We can tackle this state learning problem by defining a class of functions that take as input Gaussian channels and measurements and output their corresponding ``acceptance'' probability on a Gaussian state:
\begin{align}\label{eq:gClass}
\cF_{\rm g}(n):=\{f_{\bfm,V}&:(\bd,X,Y,\bfm',V') \\
&\mapsto P_{\rm g}(M_{\bfm',V'}|\Phi_{\bd,X,Y},\sigma_{\bfm,V})\},\nonumber
\end{align}
where different functions are effectively parametrized by different Gaussian states. 
The key ingredient in bounding the SC for this learning problem is a calculation of the pseudo-dimension of $\cF_{\rm g}(n)$:
\begin{theorem}\label{thm:pdimMeas}
${\rm Pdim}(\cF_{\rm g}(n))= O\left(n^{2}\log n\right)$.
\end{theorem}
\begin{proof}
Define the function classes 
{\begin{align}
\cF_{\rm e}:=\{f_{\bfm,V}&:(\bd,X,Y,\bfm',V') \label{eq:f1}\\
&\mapsto (\bfm_{\rm out}-\bfm')^{T}(V_{\rm out}+V')^{-1}(\bfm_{\rm out}-\bfm')\},\nonumber
\end{align}}
and
\begin{equation}
\cF_{\rm d}:=\{f_{V}:(X,Y,V')\mapsto \det{2\pi(V_{\rm out}+V')}\},\label{eq:f2}
\end{equation}
and note that
{\begin{align}
\log\cF_{\rm g}\subset\cF_{\rm e}+\frac12\log\cF_{\rm d}&:=\{f_{1}+f_{2}
\\&\mid f_{1}\in \cF_{\rm e}, f_{2}\in\frac12\log\cF_{\rm d}\},\nonumber
\end{align}}
where the sum of two function classes contains the sum of all couples of functions from each class, while $g(\cF)$ is the function class obtained by applying the function $g(\cdot)$ to every element of $\cF$. %It is well-known that $\pdim{\cF}=\pdim{g(\cF)}$ for all strictly monotonous continuous functions $g$ (). 
Hence we have
\begin{align}\label{eq:pdim_gaussian}
\pdim{\cF_{\rm g}} =\pdim{\log\cF_{\rm g}}< {\rm Pdim}\left(\cF_{\rm e}+\frac12\log\cF_{\rm d}\right)\leq O(\pdim{\cF_{\rm e}}+\pdim{\cF_{\rm d}}),
\end{align}
where we have used the properties of pseudo-dimension under composition with monotonous functions (Lemma~\ref{lemma:monot}) and sum (Lemma~\ref{lemma:sumPdim}).

We are then left to evaluate the dimension of the two function classes (\ref{eq:f1},\ref{eq:f2}), which can be done observing that they comprise functions that are polynomial in the parameters and applying Lemma~\ref{lemma:poly}.

The function class $\cF_{\rm d}$ is characterized by the $n(2n+1)$ parameters $\{V_{i,j}:i\geq j\}$. A function $f(X,Y,V')\in\cF_{\rm d}$ is related to these parameters via the determinant of the matrix $V_{\rm out}+V'\in{\rm Sym}_{2n}(\R)$; here $V'$ is the covariance matrix of the Gaussian measurement, therefore an input to the function itself, while $V_{\rm out}$ is the output covariance matrix after the Gaussian channel, therefore it depends both on the inputs $X, Y$ and on the parameters as given in \eqref{eq:output_Gaussian_channel}: $V_{\rm out} = X V X^T +Y$. The determinant implies that $f$ is a polynomial of order $2n$ in the entries of $V_{\rm out}$, which are in linear relation with the entries of $V$, i.e., the function parameters. We conclude that any $f\in\cF_{\rm d}$ is a polynomial function of order $2n$ in the real parameters $\{V_{i,j}:i\geq j\}$, and by applying Lemma~\ref{lemma:poly} with $\ell = 2n$ and $d=n(2n+1)$ we obtain
\begin{equation}\label{eq:pdim_determinant}
\pdim{\cF_{\rm d}}\leq 2(2n^{2}+n)\log(24n).
\end{equation}

The function class $\cF_{\rm e}$ instead is characterized by the same parameters of $\cF_{\rm d}$ plus the $2n$ parameters $\bfm$, for a total of $2n^2+3n$ parameters. The dependence on these parameters is via the function $(\bfm_{\rm out}-\bfm')^T(V_{\rm out}+V')^{-1}(\bfm_{\rm out}-\bfm')$.
First, we focus on the matrix inverse $(V_{\rm out}+V')^{-1}$, whose elements are given by 
\begin{equation}
    ((V_{\rm out}+V')^{-1})_{i,j} = \frac{M_{i,j}}{\det{V_{\rm out} + V'}},
\end{equation}
where $M_{i,j}$ are the minors of $V_{\rm out} + V'$, i.e., polynomials of order $2n-1$ in the matrix elements.  We can therefore further decompose this class as 
 \begin{equation}\label{eq:decomposition_exponential_class}
     \log\cF_{\rm e} = \log\cF_{\rm quad} + \log\cF_{\rm d},
 \end{equation}
 with
\begin{equation}
    \cF_{\rm quad}:=\{f_{\bfm,V}:(\bd,X,Y,\bfm',V')
\mapsto (\bfm_{\rm out}-\bfm')^{T}(M_{i,j})_{i,j=1,\cdots,2n}(\bfm_{\rm out}-\bfm')\}.
\end{equation}
Functions in $\cF_{\rm quad}$ depend on the parameters $\{V_{i,j}:i\geq j\}\cup\{m_i\}$  via polynomials of order $2n-1$, i.e., the minors, plus $2n$, i.e., the vector of mean values $\bfm_{out} = X \bfm$ (as per \eqref{eq:output_Gaussian_channel}), for a total degree of $4n-1$.  Hence, applying Lemma~\ref{lemma:poly} with $\ell=4n-1$ and $d=2n^2+3n$ we conclude that $\pdim{\cF_{\rm quad}} \leq 2(2n^2+3n)\log(12(4n-1))$. Putting this together with (\ref{eq:pdim_gaussian},\ref{eq:pdim_determinant},\ref{eq:decomposition_exponential_class})  and applying again Lemmas~\ref{lemma:monot} and \ref{lemma:sumPdim} we obtain the theorem statement:
\begin{equation}
    \pdim{\cF_{\rm g}}\leq O(\pdim{\cF_{\rm quad}} + 2\pdim{\cF_{\rm d}}) \leq O(n^2\log n).
\end{equation}
\end{proof}

If instead we are interested in learning a Gaussian measurement operator or a Gaussian channel, the corresponding function classes comprise the same functions of $\cF_{\rm g}$, though with exchanged roles of variables and parameters. For example, for measurement learning we have to consider functions 
\begin{equation}
f_{\bfm',V'}:(\bfm,V,\bd,X,Y)\mapsto P_{\rm g}(M_{\bfm',V'}|\Phi_{\bd,X,Y},\sigma_{\bfm,V}),
\end{equation}
while for channel learning
\begin{equation}
f_{\bd,X,Y}:(\bfm,V,\bfm',V')\mapsto P_{\rm g}(M_{\bfm',V'}|\Phi_{\bd,X,Y},\sigma_{\bfm,V}).
\end{equation} 
It is straightforward to see that the pseudo-dimension of these classes is still bounded by $O(n^{2}\log n)$. {\color{red} Indeed, functions parametrized by measurements coincide with those parametrized by states, while functions parametrized by channels still have a number of parameters $O(n^{2})$, i.e., the matrix elements of $X, Y$, and depend on polynomials of order $O(n)$ in these variables, e.g., the dependence of $V_{\rm out}$ on $X$ is quadratic, implying a polynomials of order $2n$ in the matrix elements of $X$.}

\subsubsection{Gaussian circuits plus photodetection}\label{sec:gpPdim}
Consider now the sample class of $n$-mode Gaussian channels plus photodetection measurements
\be
\cS_{\rm gp}(n,K):=\{(\Phi_{\bd,X,Y},M_{q})\},
\end{equation}
where $M_{q}$ is a coarse-grained photon-number measurement operator and $K$ is the maximum photon-number cutoff (see Sec.~\ref{subsubsec:gpDef}). Define then the function class of outcome probabilities obtained by applying a channel and measurement from $\cS_{\rm gp}(n,K)$ to a state in $\cC_{\rm g}(n)$:
\begin{align}\label{eq:gpClass}
\cF_{\rm gp}(n,K):=\{f_{\bfm,V}&:(\bd,X,Y,q) \\
&\mapsto P_{\rm gp}(M_{q}|\Phi_{\bd,X,Y},\sigma_{\bfm,V})\}.\nonumber
\end{align}
We can compute the pseudo-dimension of this class as well:
\begin{theorem}\label{thm:photocount}
$\pdim{\cF_{\rm gp}(n,K)} = O(n^{2}\log(n K))$.
\end{theorem}
\begin{proof}
Following the proof above, we observe that {\color{red}
\begin{align}\label{eq:pdim_photodetection}
\pdim{\cF_{\rm gp}}&\leq O(\pdim{\cF_{\rm g}}+\pdim{\cF_{\rm p}}),
\end{align}}
with $\cF_{\rm e, d}$ as before and
\begin{equation}
\cF_{\rm p}:=\left\{f_{\bfm,V}:(\bd,X,Y,q)\mapsto \sum_{\bk}q_{\bk}H_{\bk}^{(\tilde V)}(\tilde\bfm)\right\}.
\end{equation} 
{\color{red} Hence we only need to compute the pseudo-dimension of the latter class. Note that $\sum_{\bk}q_{\bk}H_{\bk}^{(\tilde V)}(\tilde\bfm)$ is a polynomial of order $n K$ in the vector $\tilde V\tilde\bfm$, by definition. In turn, both the matrix $\tilde V$ and the vector $\tilde m$ depend quadratically on $V_{\rm out}$ and $m_{\rm out}$ (see  \ref{subsubsec:gpDef} and~\cite{Dodonov1994}), which are linear functions of the parameters $\{V_{i,j}:i\geq j\}\cup\{m_i\}$, as discussed before. Hence $\tilde V\tilde\bfm$ is a polynomial of order $4$ in the parameters.
We conclude that a function $f_{\bfm,V}\in\cF_{\rm p}$ is a polynomial of order $4nK$ in its $2n^{2}+3n$ parameters and we conclude 
\begin{equation}
\pdim{\cF_{\rm p}}\leq 2(2n^{2}+3n)\log(28nK).
\end{equation}
Combining this pseudo-dimension with Theorem~\ref{thm:pdimMeas} and (\ref{eq:pdim_photodetection}) we obtain the theorem statement.}
\end{proof}
If we consider an analogous function class for channel learning, comprising functions 
\begin{equation}
f_{\bd,X,Y}:(\bfm,V,q)\mapsto P_{\rm gp}(M_{q}|\Phi_{\bd,X,Y},\sigma_{\bfm,V}),
\end{equation}
the pseudo-dimension stays the same. For measurement learning instead we have functions
\begin{equation}\label{eq:FockLearn}
f_{q}:(\bfm,V,\bd,X,Y)\mapsto P_{\rm gp}(M_{q}|\Phi_{\bd,X,Y},\sigma_{\bfm,V})
\end{equation}
which are linear in the parameters $q_{\bk}$. However the number of such parameters is $(K+1)^{n}$ and therefore the pseudo-dimension is exponential in the number of modes. 

\subsubsection{Generalized Gaussian circuits}\label{sec:ggCovNum}
Consider now the hypothesis class of $n$-mode GG states
\begin{equation}
\cC_{\rm gg}(n):=\{\rho_{\{c,\bfm,V\}_{\cI}}\}
\end{equation}
and the sample class of $n$-mode GG channels and measurements
\begin{equation}
\cS_{\rm gg}(n):=\{(\Phi_{\{c'',X,Y\}_{\cI''}},M_{\{c',\bfm',V'\}_{\cI'}})\}.
\end{equation}
Define then the function class of outcome probabilities obtained by applying a channel and measurement from $\cS_{\rm gg}(n)$ to a state in $\cC_{\rm gg}(n)$:
\begin{align}
&\cF_{\rm gg}(n):=\Bigg\{f_{\{c,\bfm,V\}_{\cI}}:(\{c'',X,Y\}_{\cI''},\{c',\bfm',V'\}_{\cI'})\nonumber \\
&\mapsto P_{\rm gg}(M_{\{c',\bfm',V'\}_{\cI'}}|\Phi_{\{c'',x,Y\}_{\cI''}},\sigma_{\{c,\bfm,V\}_{\cI}})\Bigg|\; \forall i,j,k\nonumber\\
& B_{1}\geq\abs{(\bfm_{\rm out}^{(ijk)}-\bfm_{j}')^{T}(V_{\rm out}^{(ik)}+V_{j}')^{-1}(\bfm_{\rm out}^{(ijk)}-\bfm_{j}')},\nonumber\\
& B_{2}\geq\sum_{i,j,k}\abs{c_{i} c_{j}' c_{k}''},\, B_{3}\leq \abs{\det{2\pi\left(V_{\rm out}^{(ik)}+V_{j}'\right)}}^{\frac12} \Bigg\},\label{eq:ggs}
\end{align}
where $B_{1,2,3}$ can be interpreted as mild constraints on the absolute values of the complex parameters and variables determining GG states, channels and measurements that concur to the creation of the class. Note that these constraints are connected with the degree of non-Gaussianity of the GG state or measurement, e.g., $B_{2}$ controls the number of components in the sum, while $B_{1,3}$ control the strength of the various terms in the sum; see Appendix~\ref{app:suffBounds} for a more explicit interpretation of the first constraint and Appendix~\ref{app:explicitGG} for an evaluation of these constraints for relevant classes of GG states, i.e., GKP, cat and Fock states. 

{\color{red} It is easy to show that the pseudo-dimension of $\cF_{\rm gg}$ is infinite, due to the potentially infinite number of terms in the complex-combination. Yet, we can show that its covering number is still finite, provided that $B_{1,2,3}$ are finite and that the coefficients of the GG measurements and channels we are learning are fixed. This further constraint is also mild, in the sense that it does not trivialize the learning problem. It amounts to fixing some properties of the setup of the device, i.e., evolution and measurement, that is employed to test the unknown state. 

For example, for the most relevant cases in~\cite{Bourassa2021}, fixing the GG measurement coefficients amounts to \emph{fixing a specific projection state} $\dketbra{\psi}$ in the GG class, e.g., a well-defined GKP, cat or Fock state; an admissible class of measurement operators will thus have the form
\begin{equation}\label{eq:admissible_gg_measurements}
    M(U) = U^\dag \dketbra{\psi} U,
\end{equation}
where $U$ are Gaussian unitaries labelling the measurement outcome, e.g., displacement operators $U=D(\br)$. Generalizing this construction, another example are measurement operators of the form $ M(U) = U^\dag M_0 U$, where now $M_0$ is a fixed GG operator, not necessarily projective nor rank-1.

Instead, fixing GG channel coefficients amounts to consider conditional dynamics based on partial measurements with \emph{fixed non-Gaussian components}, e.g.,
\begin{equation}
    \Phi(\rho) = \sum_i \Phi^{(i)}_{A}({\rm Tr}_{S}[M_{A}^{(i)}V_{SA}(\rho_S \otimes \sigma_A)V_{SA}^\dag]),
\end{equation}
where $\sigma$ and $\{M^{(i)}\}_{i}$ are  respectively a GG state and measurement on an ancilla register $A$, coupled to the input system $S$ via a Gaussian unitary $V$, and $\Phi^{(i)}$ are Gaussian quantum channels acting on the system conditional on the measurement outcome $i$. An admissible class of GG channels of this form is obtained by fixing $\{M^{(i)}\}_i$ and $\sigma$, while varying $V$ and $\{\Phi^{(i)}\}_i$.
This is fine for most practical applications, where non-Gaussian evolution is based on a well-defined measurement and ancilla, e.g., Fock damping, squeezed ancilla-assisted gates and the GKP T-gate mentioned in Ref.~\cite{Bourassa2021}.
}
\begin{theorem}\label{theorem:covNumBound}
Let $\{c_{j}'\}_{j\in\cI'}$, $\{c_{k}''\}_{k\in\cI''}$ be fixed and $\cN_{2}(\epsilon,\cF_{\rm gg},k)$ be the covering number with respect to the Euclidean distance. Then
\begin{align}\label{eq:covGG}
\log \cN_{2}(\epsilon,\cF_{\rm gg},k)\leq \left\lceil\left(\frac{2B}{\epsilon}\right)^{2}\right\rceil d\log\frac{2ek \tilde B}{d\epsilon},
\end{align}
with $d=O(n^{2}\log n)$, $B=\frac{B_{2}}{B_{3}}$ and $\tilde B=O(\log(B_{1}B))$.
\end{theorem}
\begin{proof}
Since the outcome probability \eqref{eq:gg} is a real function we can rewrite it as $P_{\rm gg}={\rm Re}P_{\rm gg}$, obtaining
%\begin{align}
%P_{\rm gg}&=\sum_{i,j,k}{\rm Re}\left[c_{i}c_{j}'c_{k}''\frac{e^{-R_{ijk}}}{M_{ijk}}e^{-\i(I_{ijk}+A_{ijk})}\right]\\
%&=\sum_{i,j,k} \frac{e^{-R_{ijk}}}{M_{ijk}}[p_{ijk} \cos(I_{ijk}+A_{ijk}) \\
%&+ q_{ijk} \sin(+I_{ijk}+A_{ijk})],\nonumber
%\end{align}
\begin{align}
P_{\rm gg}&=B_{2}\sum_{ijk}{\rm Re}\left[p_{ijk}\frac{e^{-R_{ijk}}}{M_{ijk}}e^{-\i(I_{ijk}+A_{ijk}-\theta_{ijk})}\right]\\
&=B_{2}\sum_{ijk} p_{ijk}\frac{e^{-R_{ijk}}}{M_{ijk}}\cos(I_{ijk}+A_{ijk}-\theta_{ijk})
\end{align}
where in the first equality we have defined the real and imaginary parts of the Gaussian's exponent,
{\begin{align}
 &R_{ijk} :=\frac12{\rm Re}\left[(\bfm_{\rm out}^{(ijk)}-\bfm_{j}')^{T}(V_{\rm out}^{(ik)}+V_{j}')^{-1}(\bfm_{\rm out}^{(ijk)}-\bfm_{j}')\right],\\
  &I_{ijk} :=\frac12{\rm Im}\left[(\bfm_{\rm out}^{(ijk)}-\bfm_{j}')^{T}(V_{\rm out}^{(ik)}+V_{j}')^{-1}(\bfm_{\rm out}^{(ijk)}-\bfm_{j}')\right],
\end{align}}
the radial and angular parts of the Gaussian's normalization,
\begin{align}
&M_{ijk}:=\left|\det{2\pi\left(V_{\rm out}^{(ik)}+V_{j}'\right)}\right|^{\frac12},\\
&A_{ijk}:=\frac12\arccos{\frac{{\rm Re}\,\det{2\pi\left(V_{\rm out}^{(ik)}+V_{j}'\right)}}{M_{ijk}}},
\end{align}
and the radial and angular parts of the product of the complex combination coefficients,
\begin{equation}
p_{ijk}:=\frac{|c_{i}c_{j}'c_{k}''|}{B_{2}},\quad  \theta_{ijk}:=\arccos\frac{{\rm Re}(c_{i}c_{j}'c_{k}'')}{B_{2}\, p_{ijk}},
\end{equation}
the former rescaled so that $\sum_{i,j,k}p_{ijk}\leq1$. %, which satisfy the normalization $\sum_{i,j,k}(p_{ijk}+\i q_{ijk})\leq 1$ and hence $\sum_{i,j,k}p_{ijk}\leq 1$ and $\sum_{i,j,k}q_{ijk}=0$. 
Let us then define the function classes
{\begin{align}
&\cF_{\rm exp}:=\left\{f_{(\bfm,V)_{i}}:((X,Y)_{k},(\bfm',V')_{j})\mapsto \frac{e^{-R_{ijk}}}{M_{ijk}}\right\},\\
&\cF_{\rm im}:=\left\{f_{(\bfm,V)_{i}}:((X,Y)_{k},(\bfm',V')_{j})\mapsto I_{ijk}\right\},\\
&\cF_{\rm ang}:=\left\{f_{(\bfm,V)_{i}}:((X,Y)_{k},(\bfm',V')_{j})\mapsto A_{ijk}\right\},\\
&\cF_{\rm const}:=\{\theta_{ijk}\}, \quad \cF_{\rm trig}:=\cos(\cF_{\rm im}+\cF_{\rm ang}+\cF_{\rm const}), 
\end{align}} 
where the latter makes use of the sum and composition operations defined before. In particular, note that each value of $\theta_{ijk}$ determines a different function of $\cF_{\rm trig}$ and that we can assume without loss of generality $\theta_{ijk}\in[0,2\pi)$. 

{\color{red} Putting all these definitions together, we conclude that $\cF_{\rm gg}\subseteq B_{2}\mathfrak{C}(\cF_{\rm exp}\cdot\cF_{\rm trig})$, where $\mathfrak C$ represents the convex-hull operation, and the product of two function classes contains all the products of two functions, one from each class. } We stress that the convex-hull coefficients $p_{ijk}$ can depend exclusively on the parameters of the function class and not on the input variables, which obliges us to fix the channel and measurement coefficients in the assumptions of the Theorem.  Then we can apply Lemmas~\ref{lemma:convHull},\ref{lemma:covProd} to prove that
\begin{align}\label{eq:covN1}
&\log\cN_{2}(\epsilon,\cF_{\rm gg},k)\leq\left\lceil\left(\frac{2B_{2}}{\epsilon B_{3}}\right)^{2}\right\rceil\log\cN_{2}\left(\frac{\epsilon}{2B_{2}},\cF_{\rm exp}\cdot\cF_{\rm trig},k\right)\nonumber\\
&\leq \left\lceil\left(\frac{2B_{2}}{\epsilon B_{3}}\right)^{2}\right\rceil\log(\cN_{2}\left(\tilde\epsilon,\cF_{\rm exp},k\right) \cN_{2}\left(\tilde\epsilon,\cF_{\rm trig},k\right)).
\end{align}
where $\tilde\epsilon=\frac{\epsilon B_{3}}{2 B_{2} (1+B_{3})}$, and we have used the fact that functions in $\cF_{\rm exp}$ are bounded by $1/B_{3}$, while those in $\cF_{\rm trig}$ by $1$.
Note furthermore that, by Lemma~\ref{lemma:lipschitz} and since $\cos(x)$ is $1$-Lipschitz, it holds $\cN_{2}\left(\tilde\epsilon,\cF_{\rm trig},k\right)\leq\, \cN_{2}\left(\tilde\epsilon,\cF_{\rm im}+\cF_{\rm ang}+\cF_{\rm const},k\right)$. We then proceed to evaluate the pseudo-dimensions of $\cF_{\rm exp}$, $\cF_{\rm im}$, $\cF_{\rm ang}$ and $\cF_{\rm const}$, which in turn will allow to bound their covering numbers via Eq.~\eqref{eq:covPdim}.

For the first class we can follow the same steps of Theorem \ref{thm:pdimMeas}, observing that $\log\cF_{\rm exp}\subset {\rm Re} \tilde\cF_{\rm e} +\frac14 \log|\tilde\cF_{\rm d}|^{2}$, with $\tilde\cF_{\rm e, d}$ generalizing the classes of Eqs. ~(\ref{eq:f1},\ref{eq:f2}) to complex-valued parameters and variables, without symmetry constraints. This implies that the functions in $\tilde\cF_{\rm d}$ and $\tilde\cF_{\rm e}$ are polynomials in, respectively, $8n^{2}$ and $2(4n^{2}+2n)$ real parameters, of the same order as $\cF_{\rm d}$ and $\cF_{\rm e}$. Furthermore, note that the degree of a complex polynomial does not change by taking its real part, while it doubles by taking its radial component squared. Therefore we conclude that the functions in ${\rm Re}\tilde\cF_{\rm e}$ are polynomials of degree $4n+1$, implying 
\be
\pdim{{\rm Re}\tilde\cF_{\rm e}}\leq  4(4n^{2}+2n)\log(12(4n+1)),
\end{equation} 
while the functions in $|\tilde\cF_{\rm d}|^{2}$ are polynomials of degree $4n$, implying
\be
\pdim{|\tilde\cF_{\rm d}|^{2}}\leq 16n^{2}\log(48n),
\end{equation} 
so we conclude that $\pdim{\cF_{\exp}}=O(n^{2}\log n)=:d$.
We can now apply Eq.~\eqref{eq:covPdim}, observing that the functions in $\cF_{\exp}$ have range $[0,\frac{1}{B_{3}}]$, obtaining
\begin{equation}\label{eq:covN2}
\log \cN_{2}\left(\tilde\epsilon,\cF_{\rm exp},k\right)\leq d \log\frac{2ek}{d\tilde\epsilon B_{3}}.
\end{equation}

We proceed similarly to bound the covering number of $\cF_{\rm im}+\cF_{\rm ang}+\cF_{\rm const}$. Firstly, it is well-known that $\pdim{\cF_{\rm const}}=1$; secondly, observe that $\cF_{\rm im}=\rm{Im}\tilde\cF_{\rm e}$, hence we conclude immediately $\pdim{\cF_{\rm im}}=\pdim{{\rm Re}\tilde\cF_{\rm e}}$. The class $\cF_{\rm ang}$ instead requires a bit more work: using the fact that $\arccos$ (restricted to the interval $[-1,1]$) and $\log$ are strictly monotonous functions we have, applying again Lemma~\ref{lemma:monot}, that
\begin{align}
\pdim{\cF_{\rm ang}}&\leq\pdim{{\rm Re}\tilde\cF_{\rm d}\cdot|\tilde\cF_{\rm d}|^{-\frac12}}\\
&\leq\pdim{\log{\rm Re}\tilde\cF_{\rm d}-\frac14\log|\tilde\cF_{\rm d}|^{2}}\\
&\leq O(\pdim{{\rm Re}\tilde\cF_{\rm d}}+\pdim{|\tilde\cF_{\rm d}|^{2}}).
\end{align}
It is then easy to show that $\pdim{{\rm  Re}\tilde\cF_{\rm d}}=16n^{2}\log(24n)$ and conclude that $\pdim{\cF_{\rm ang}}=O(n^{2}\log n)$ as well. 
Finally, we apply Eq.~\eqref{eq:covPdim}, observing that for any function $f\in(\cF_{\rm im}+\cF_{\rm ang}+\cF_{\rm const})$ it holds $|f|\leq |I_{ijk}+A_{ijk}+\theta_{ijk}| \leq B_{1} + 3\pi$, so that
\begin{equation}\label{eq:covN3}
\log \cN_{2}\left(\tilde\epsilon,\cF_{\rm trig},k\right)\leq d \log\frac{2e(B_{1}+9)k}{d\tilde\epsilon}.
\end{equation}
By combining Eqs.~(\ref{eq:covN1},\ref{eq:covN2},\ref{eq:covN3}) we obtain the claim.
\end{proof}
%Note that, unlike the case of Gaussian states, learnability for GG states and processes requires a restriction on all the parameters concurring to define the GG decomposition: (i) the mean values and covariance matrices; (ii) the measurement coefficients, thus bounding the support of the corresponding measurement operators. In particular, note that the first condition can be always satisfied by upper-bounding the mean values $\norm{\bs_{i}}, \norm{\br_{j}}$ and lower-bounding the smallest eigenvalue of $(W_{i}+V_{j})(W_{i}+V_{j})^{\dag}$, as we show in .

%Additionally, one can include the set of Gaussian channels, which act as usual on the complex-valued Gaussians, as well as conditional generalized Gaussian dynamics given by the Gaussian interaction with a set of GG ancillary modes, followed by a GG measurement. Since the latter case is a straightforward though cumbersome generalization, here we will restrict to standard Gaussian dynamics without ancillary modes. 

\subsection{SC bounds for CV state, measurement and channel learning}\label{sec:learnProofs}
The quantities computed in Sec.~\ref{sec:compBounds} can be straightforwardly employed to obtain SC bounds for the state, measurement and channel learning problems described in Sec.~\ref{sec:stateLearn}. The simplest case is that where only Gaussian or GP circuit components are involved:
\begin{proof}(Theorem~\ref{thm:statelearnedot} for Gaussian or GP circuits)
The function classes $\cF_{\rm g}$, Eq.~\eqref{eq:gClass}, and $\cF_{\rm gp}$, Eq.~\eqref{eq:gpClass}, arise when trying to learn Gaussian states with Gaussian or GP channels and measurements. The SC of these tasks can then be evaluated by applying Theorem~\ref{thm:genLearn} to the classes $\cF_{\rm g, gp}$, concluding that the following number of samples is sufficient for learning
\begin{align}
&T_{\rm g}=O\left(\frac{1}{\nu^{4}\epsilon^{2}}\left(n^{2}\log(n)\log^{2}\frac{1}{\nu\epsilon} + \log\frac1\delta\right)\right),\\
&T_{\rm gp}=O\left(\frac{1}{\nu^{4}\epsilon^{2}}\left(n^{2}\log(nK)\log^{2}\frac{1}{\nu\epsilon} + \log\frac1\delta\right)\right),
\end{align}
where we employed $\fdim{\cF}{\cdot}\leq\pdim{\cF}$ and the pseudo-dimension bounds of Theorems~\ref{thm:pdimMeas},\ref{thm:photocount}.
\end{proof}
An analogous proof technique can be employed for channel or measurement learning. 
Let us note that these scalings are perhaps not particularly surprising, since tomography of Gaussian states and processes also requires a number of measurements scaling polynomially with the number of modes. Indeed, for an $n$-mode Gaussian state $\rho_{\bfm,V}$, the mean $\bfm$ can be estimated with error $\epsilon$ via $O(n \epsilon^{-2})$ single-mode homodyne measurements, while $V$ requires additional $O(n^{2}\epsilon^{-2})$ homodyne measurements. Therefore the SC complexity for learning a Gaussian state using Gaussian or photodetection measurements seems comparable with that of performing full tomography. 

On the other hand, the surprising result is that a similar scaling is obtained when the state is not Gaussian, but rather GG, as proved below:
\begin{proof}(Theorem~\ref{thm:statelearnedot} for GG circuits)
In this case we would need to apply Theorem~\ref{thm:genLearn} to the function class $\cF_{\rm gg}$~\eqref{eq:ggs}. However, we only have a bound on its covering number and not on its fat-shattering or pseudo-dimension. 

We then need to repeat the proof of~\cite[Theorem 12]{Bartlett1998}, upon which the SC of Theorems~\ref{thm:genLearn},\ref{thm:genPred} relies; indeed, we can apply directly our covering number bound, right after the following statement is proved (Theorem 12 thereof):
{\begin{align}
{\rm Pr}_{\{(x_{t},b_{t})\}\sim (Q\cdot P)^{\times T}}&\left[\left|\frac1T\sum_{t=1}^{T}\cL(x_{t},b_{t})-\mathbb{E}\cL(x,b)\right|\geq\epsilon'\right]\nonumber\\
&\leq4\,\cN_{\tilde 1}\left(\frac{\epsilon'}{2}-\alpha,\cF,2T\right)e^{-\alpha^{2} \frac T2},\label{eq:bartLearn}%{1-2e^{-2\alpha^{2}k}}
\end{align}}
%{\begin{align}
%{\rm Pr}_{\{(x_{t},b_{t})\}\sim (Q\cdot P)^{\times 2T}}&\left[\left|\hat R(f,\{(x_{t},b_{t})\}_{t=1}^{2T})-R(f,Q\cdot P)\right|\geq\epsilon'\right]\nonumber\\
%&\leq4\,\cN_{\tilde 1}\left(\frac{\epsilon'}{2}-\alpha,\cF,2T\right)e^{-\alpha^{2} \frac T2},\label{eq:bartLearn}%{1-2e^{-2\alpha^{2}k}}
%\end{align}}
where $\alpha=\lceil1/(\epsilon' \kappa)\rceil^{-1}$, $0<\kappa\leq1/4$. Note that we can take $\cL(x,b)=(f(x)-b)^{2}$ for all $f\in\cF_{\rm gg}$ to recover our learning problem (as shown in~\cite[Theorem 8.1]{Aaronson2007}).
%and we have defined the empirical risk $\hat R(f,\{(x_{t},b_{t})\}_{t=1}^{T}):=\frac1T\sum_{t=1}^{T}\cL(f,(x_{t},b_{t}))$ and the expected risk $R(f,Q\cdot P):=\mathbb{E}_{(x,b)\sim Q\cdot P}\cL(f,(x,b))$

Importantly, the covering number $\cN_{\tilde 1}$ is computed with respect to a rescaled $1$-norm distance, $ d_{\tilde 1}(\bx,\by):=\frac1T\sum_{t=1}^{T}|x_{t}-y_{t}|$, for $\bx,\by\in \cX^{\times T}$, rather then the Euclidean distance $d_{2}(\bx,\by)$ used in Eq.~\eqref{eq:covGG}.
We can bound this covering number as follows:
\begin{align}
\cN_{\tilde 1}\Big(\frac{\epsilon'}{2}-\alpha,&\cF,2T\Big) \leq \cN_{\tilde 1}\left(\frac{\epsilon'}{4},\cF,2T\right) \leq\cN_{2}\left(\frac{\epsilon'}{4},\cF,2T\right) \nonumber\\
&\leq \cN_{2}\left(\frac{\epsilon'}{4},\cF_{gg},2T\right)^{2}\\
&\leq \exp\left(2\left\lceil\left(\frac{8B}{\epsilon'}\right)^{2}\right\rceil d\log\frac{16eT \tilde B}{d\epsilon'}\right).
\end{align}
The first inequality above follows from the fact that the covering number is a decreasing function of the covering's size and that $\epsilon'/2-\alpha\geq\epsilon'/4$;  the second one instead follows from the fact that covering numbers increase when computed with respect to a larger distance and clearly $d_{\tilde 1}\leq d_{2}$; the third inequality follows from $\cF=(\cF_{gg})^{2}$ and Lemma~\ref{lemma:covProd}; finally, the last one follows from Theorem~\ref{theorem:covNumBound} above, with $d=O(n^{2}\log n)$, $B=\frac{B_{2}}{B_{3}}$ and $\tilde B=O(\log(B_{1}B))$.

%{\color{red} We might as well use \cite[Theorem 12.11]{Anthony1999} to bound the fat-shattering dimension with covering numbers, inside the sample complexity. We have
%\begin{align}
%\fdim{\cF}{\epsilon}&\leq 8\log\cN_{2}\left(\frac{\epsilon}{16},\cF,m\right)\\
%&\leq O\left(\frac{n^{2}}{\epsilon^{2}}\log n\left(\log\frac{m C B}{\epsilon}\right)\right)
%\end{align}
%but we need $m\geq\fdim{\cF}{\epsilon/4}$, which is not easy to solve.}

Using this bound and~\cite[Lemma 18]{Bartlett1998}, with $y_{1}=4$, $y_{2}=2d\left\lceil\left(\frac{8B_{2}}{\epsilon' B_{3}}\right)^{2}\right\rceil$, $y_{3}=\frac{16eB}{d\epsilon'}$ and $y_{4}=\frac{\alpha^{2}}{2}$, we can now obtain a sufficient $T$ for learnability to hold, i.e., such that the right-hand side of Eq.~\eqref{eq:bartLearn} is bounded by $\delta$:
\begin{equation}\label{eq:final_onslaught}
T_{\rm gg}=O\left(\frac{d B^{2}}{\epsilon'^{4}}\log\frac{B \tilde B}{\epsilon' d }+\frac{1}{\epsilon'^{2}}\log\frac1\delta\right),
\end{equation}
which implies the Theorem's statement for $\epsilon'=\nu^{2}\epsilon$.
\end{proof}
%We expect that the additional factor $1/\epsilon^{2}$ in the first term of \eqref{eq:compGGS} can be turned into $\log1/\epsilon$ by proving a variant of Lemma~\ref{theorem:covNumBound} directly for $\cN_{\tilde 1}$.
Also in this case, the proof can be easily extended to the problems of channel and measurement learning. Importantly, note that a crucial condition for applying Theorem~\ref{theorem:covNumBound} is that the GG combination coefficients are fixed for the sample class. Hence, for example, we can guarantee efficient learnability of arbitrary GG states only with respect to GG channels and measurements with fixed coefficients. This covers interesting cases, e.g.: (i) Gaussian channels and measurements, (ii) channels generated by a Gaussian interaction with a fixed GG ancillary state and (iii) measurements generated by projecting on a fixed GG state after applying a Gaussian unitary.

{\color{blue} The resulting SC is $O(n^{2}\epsilon^{-4})$, obtained by inserting in \eqref{eq:final_onslaught} $\epsilon'=O(\epsilon)$ and $d=O(n^2)$ and hiding logarithmic terms. Thus, the SC of learning this non-Gaussian class of states }is only $\epsilon^{-2}$-larger than the SC of learning Gaussian states or performing Gaussian state tomography. Furthermore, it is clear that Gaussian state tomography does not work at all in this case, since it allows a correct reconstruction of just the mean and covariance matrix of the unknown GG state, which are not sufficient to reconstruct its measurement statistics.

\subsection{SC bounds for CV circuit compilation}\label{sec:predProofs}
The techniques employed in Sec.~\ref{sec:learnProofs} can also be applied to bound the SC of learning the optimal family of quantum circuits for a discrimination or synthesis task, as described in Sec.~\ref{sec:circuitTrain}. However, there are two major differences: (i) we have to consider a prediction version of the problem, i.e., applying Theorem~\ref{thm:genPred}; (ii) we have to consider function classes of the form $\cF\circ\cE$, obtained via composition of an encoding class $\cE:=\{x\mapsto\bgamma(x)\}$, that maps input variables into circuit parameters, and a probability-valued function class, that maps circuits into outcome probabilities as before, i.e., $\cF_{\rm g, gp, gg}$.
{\color{blue} Indeed, in the case of task learnability, differently form the learning setting of Sec.~\ref{sec:stateLearn}, one can allow the freedom to optimize also a classical function that maps the input random variable $x$ to a vector  $\bgamma(x)$, which determines the Gaussian and non-Gaussian parameters of the hypotheses. For example, suppose we want to learn an optimal Gaussian measurement $\{M_x\}$ for a discrimination problem. Then, each $M_x$ will be parametrized in terms of a mean-value vector and covariance matrix $(\bfm_x, V_x)$, but the specific dependence of these on the classical input $x$ still requires to be specified. One can do so via the classical of functions $\bgamma(x)$, whose components give the components of $\bfm_x$ and matrix elements of $V_x$. In this way, one can effectively allow a compilation not only of the quantum circuit performing the measurement, but also of the classical post-processing of the emasurement outcomes.

Given the function classes $\cF$ and $\cE$, }if the resulting function class $\cF\circ\cE$ also has bounded pseudo-dimension or covering number, then we obtain learnability also in this case. Unfortunately, the properties of such complexity measures under composition of classes are not simple to apply in our case, and hence we cannot provide a formula for general $\cE$. Still, we can prove learnability for the specific class of polynomial encoding functions: 
\begin{proof}(Theorem~\ref{thm:circuitTrainTot})
The SC of learning Gaussian or GP circuits for discrimination and synthesis can be computed by applying Theorem~\ref{thm:genPred} to the function classes $\cF_{\rm g, gp}\circ\cE_{\ell}$, where
\be
\cE_{\ell}:=\left\{x\mapsto\bgamma(x)|\gamma_{i}(x)={\rm poly}_{\ell}(x)\forall i\right\}
\end{equation}
is a class of encoding functions that are polynomials of order $\ell$ in the input $x$. 

For their pseudo-dimension we can proceed as in Theorems \ref{thm:pdimMeas},\ref{thm:photocount}, up until the point where we need to evaluate the pseudo-dimension of $\cF_{\rm e, d, p}\circ \cE_{\ell}$, which are still classes of polynomials in their defining parameters. In particular, note that the encoding $\cE_{\ell}$ contributes a total of $\ell$ additional parameters for each input of $\cF$, which has always $O(n^{2})$ inputs, and that each of these parameters appears with a linear dependence in the encoding function. The only thing we have to be careful about is counting the order of the original functions in the inputs. 
For example consider the class $\cF_{\rm d}$, which is a polynomial of order $2n$ in the $2n(2n+1)$ parameters $\{V_{i,j}:i\geq j\}$, as before, but also a polynomial of the same order in the variables $\{V'_{i,j}:i\geq j\}$ and a polynomial of order $4n$ in the variables $\{X,Y\}$. However, when we compose it with the encoding class, $\cF_{\rm d}\circ\cE_{\ell}$, each of these variables becomes a parametrized function from $\cE_{\ell}$, in the form of a polynomial of order $1$ in $\ell$ additional parameters. Thus we conclude that the functions in $\cF_{\rm d}\circ\cE_{\ell}$ are polynomials of order $O(n)$ in $O(\ell n^{2})$ parameters. Similar considerations hold for the other two classes.
Therefore, the overall modification to the pseudo-dimensions is the presence of $O(\ell n^{2})$ parameters and we obtain 
\begin{align}
&T_{\rm g}=O\left(\frac{1}{\nu^{4}\epsilon^{2}}\left(\ell n^{2}\log(n)\log^{2}\frac{1}{\nu\epsilon} + \log\frac1\delta\right)\right),\\
&T_{\rm gp}=O\left(\frac{1}{\nu^{4}\epsilon^{2}}\left(\ell n^{2}\log(nK)\log^{2}\frac{1}{\nu\epsilon} + \log\frac1\delta\right)\right).
\end{align}
In the case of GG circuits, the proof is analogous but we further need to restrict to encodings that map the input $x$ to GG circuit parameters with fixed combination coefficients. 
\end{proof}

Finally, note that the class of polynomial encodings covers several cases of interest, since polynomials of arbitrary degree can be used to approximate any desired function. 
We believe that similar results can be obtained for larger classes of encoding functions, e.g., neural networks, but leave this proof open for future works.
%\cite[Lemma 14.3]{Anthony1999} establishes that the $\infty$-covering number of a composition of function classes is the product of the $\infty$-covering numbers of the two function classes with respect to the metrics of their spaces. Therefore if we want to learn an optimal circuit $C(x)$ for a given task on average wrt $x$, we have a function class $G\circ F$ with $F\ni f:x\mapsto(\rho,\Phi,M)$ and $G\ni g:(\rho,\Phi,M)\mapsto \tr{C(x)}$. We can bound the covering number for the latter class with our methods and then obtain a SC that depends on the $\pdim F$.

{\color{blue}
\section{Discussion of learnability for GG states}
\subsection{Sufficient conditions for learnability of GG states}\label{app:suffBounds}
Here we manipulate further the $B_{1}$-constraint for GG learnability, clarifying its physical significance. For a generic complex-valued vector $\bx$ and matrix $A$ we have
\begin{align}
|\bx^{T}A\bx|&=\ave{\bx^{*},A\bx}\leq ||\bx^{*}||\cdot||A\bx||\\
&\leq ||\bx||^{2}\cdot||A||_{\infty} =  ||\bx||^{2}\cdot\lambda_{\max}(|A|),
\end{align}
where $\ave{\bs,\br}=\bs^{\dag}\br$ is the standard Euclidean inner product over the complex numbers and $||\bs||=\sqrt{\ave{\bs,\bs}}$ the induced norm. Futhermore, the first inequality is obtained via the Cauchy-Schwartz inequality, while the second one from the definition of the operator norm $||A||_{\infty}:=\sup_{\bx}||A\bx||/||\bx||$. Applying this result with $\bx=\bfm_{\rm out}^{(ijk)}-\bfm_{j}'$ and $A=(V_{\rm out}^{(ik)}+V_{j}')^{-1}$ we obtain
\begin{equation}
I_{ijk}\leq ||\bfm_{\rm out}^{(ijk)}-\bfm_{j}'||^{2}\cdot\lambda_{\min}(|V_{\rm out}^{(ik)}+V_{j}'|)^{-1}.
\end{equation}
If the mean values $||\bfm_{\rm out}^{(ijk)}||, ||\bfm_{j}'||$ are bounded for each $i,j$, then the first term in the product above is also bounded; this condition implies that the SC increases when trying to learn a GG state with larger mean. 
As for the second term, it is upper-bounded whenever the minimum eigenvalue of  $|V_{\rm out}^{(ik)}+V_{j}'|$ is lower-bounded. For example, using Weyl's inequality we obtain the bound
\begin{align}
\lambda_{\min}(&|V_{\rm out}^{(ik)}+V_{j}'|)^{2}\geq \lambda_{\min}(V_{j}'V_{j}'^{\dag})\\
&+\lambda_{\min}(V_{\rm out}^{(ik)}V_{\rm out}^{(ik)\dag})+\lambda_{\min}(V_{j}'V_{\rm out}^{(ik)\dag}+V_{\rm out}^{(ik)}V_{j}'^{\dag}).\nonumber
\end{align}
This condition is related to how much the single components of the GG state violate the uncertainty principle, since we know that the latter bounds the eigenvalues of the covariance matrix of a quantum state away from zero.
Hence, the more unphysical the single GG state components are, the larger the SC required to learn it. Finally, note that a lower bound on $\lambda_{\min}(|A|)$ implies also a lower bound on $|\det{A}|$ and therefore one could absorb the third constraint $B_{3}$ into the first one.

\subsection{Learnability constraints for useful classes of GG states}\label{app:explicitGG}
Here we analyze in more detail the significance of the GG learnability constraints on specific kinds of non-Gaussian states that can be approximated by the GG class. We employ the analysis of~\cite[Section IV]{Bourassa2021} and estimate the explicit dependence of $B_{1,2,3}$ in Theorem~\ref{theorem:covNumBound} on the states' parameters. Specifically, we need to obtain a lower bound on $\lambda_{\min}(|V_{i}|)$ and upper bounds on $||\bfm_{i}||$, $\sum_{i}|c_{i}|$, since then the constants have to scale at least as $B_{1}=O(||\bfm_{i}||\lambda_{\min}(|V_{i}|)^{-1})$, $B_{2}=O(\sum_{i}|c_{i}|)$ and $B_{3}=\Omega(\lambda_{\min}(|V_{i}|)^{\frac12})$ with the parameters. Clearly, there is an extra dependence on the channel and measurement parameters that we don't analyze here, though we assume that these have bounded parameters and non-zero displacements.  %Note that these are normalized states, so we can take $C=1$ even if we are using them as measurements. 
\begin{description}
\item[Finite-energy GKP states] the relevant parameters are the damping factor $\epsilon$ and the lattice size $L$, i.e., the maximum position or momentum allowed when truncating the state; the total number of terms in the sum is $L^{2}$. Then, from \cite[Eq.~43]{Bourassa2021} we obtain $\lambda_{\min}(V_{i})\gtrsim\epsilon$, $\norm{\bfm_{i}}\lesssim(1-2\epsilon)2 L$ at first order in $\epsilon$, i.e., when the states approximate GKP states well, and $\sum_{i}|c_{i}|\leq L^{2}||\bfm_{i}||\lesssim(1-2\epsilon)2L^{3}$. Therefore 
\be
B_{1}=O(L\epsilon^{-1}),\; B_{2}=O(L^{3}),\; B_{3}=\Omega(\epsilon^{\frac12}).
\end{equation}
%For the alternative squeezed-comb states \cite[Eqs.~49-52]{Bourassa2021}, we have $\lambda_{\min}(|V_{i}|)\sim e^{-2r}$ the SC similarly scales logarithmically with the lattice size $N$, the maximum displacement energy $d^{2}$ and linearly with the squeezing parameter $r$, in the large-$r$ limit where the the states are a good approximation of ideal GKP states. %Interestingly, using the alternative decomposition \cite[Eq.~44]{Bourassa2021} with complex-valued means one can get an $\epsilon$-independent bound at first order.
\item[Cat states] these can be expressed exactly as a GG state with four components \cite[Eq.~55]{Bourassa2021}, with $V_{i}=\1$, $||\bfm_{i}||\lesssim|\alpha|$ and $\sum_{i}|c_{i}|=1$ for the ``+'' cat state, while $\sum_{i}|c_{i}|=(1+e^{-2\abs{\alpha}^{2}})(1-e^{-2\abs{\alpha}^{2}})^{-1}$ for the ``-'' cat state. Here $\alpha$ is the amplitude of the coherent states and a larger $|\alpha|$ identifies a ``more macroscopic'' superposition; hence the region of interest is that of large $|\alpha|$, where $\sum_{i}|c_{i}|\sim 1$ also for the ``-'' cat state. In conclusion, for large-energy cat states, i.e., $\abs{\alpha}^{2}\gg 1+O(e^{-2\abs{\alpha}^{2}})$, we have
\be
B_{1}=O(|\alpha|),\; B_{2}=O(1),\; B_{3}=\Omega(1).
\end{equation}
\item[Fock states] for a cutoff $K$ on the maximum photon number, we can approximate Fock states using squeezed GG states~\cite[Eq.~61]{Bourassa2021}, governed by a parameter $r<\frac{1}{\sqrt{K}}$, which approaches exact Fock states in the limit $r\rightarrow0$. Then we have $\lambda_{\min}(|V_{i}|)\lesssim 1+O(K r^{2})$, $||\bfm_{i}||=0$ and $\sum_{i}|c_{i}|\lesssim K e^{O(r^{-2})+O(K\log r^{-2})}$. Therefore $B_{1}=B_{3}=O(1)$ while
\be
B_{2}=e^{O(r^{-2})+O(K\log r^{-2})+O(\log K)},
\end{equation}
which means that photon-number states can be learned with exponential difficulty in the photon-number $K$ and the approximation precision $r^{-2}$, consistently with the predictions of Eq.\eqref{eq:FockLearn}. 
\end{description}
}

\section{Conclusions}\label{sec:conclusions}
The interpretation of our results relies on a good understanding of two key terms: (i) \emph{learnability} means that we can find a quantum circuit that approximates an unknown process using only samples of random measurements performed on the latter; (ii) \emph{efficiency} means that the number of samples needed to reach a good approximation is polynomial in the circuit's characteristics, e.g. its size.

In this article, we established efficient learnability for CV quantum circuits comprising Gaussian and certain non-Gaussian components. These circuits describe, for example, photonic processors at the state-of-the-art, as well as more futuristic ones based on approximate cat or GKP states. Moreover, our results are platform-independent and hence also apply to other infinite-dimensional quantum circuits based, for example, on mechanical or cavity-coupled resonators. 

Our learnability result for CV circuits applies to a variety of information-processing tasks of general interest: (i) approximating the measurement statistics of an unknown state or process and (ii) approximating the optimal circuit for state or process discrimination and for pure-state synthesis. 

Furthermore, we believe that our methods can be extended to prove efficient learnability for CV circuits comprising CV states \emph{with bounded $P$-function} and Gaussian or GG channels and measurements. Indeed, every CV state can be written as an affine combination of Gaussian coherent states via the $P$-representation~\cite{serafiniBOOK}; hence the Gaussian or GG measurement statistics on this state is itself an affine combination of probabilities of the form Eqs.~(\ref{eq:g},\ref{eq:gg}) and we can bound the covering number of these classes using the same methods of Theorem~\ref{theorem:covNumBound}.

Finally, our results are a starting point for a systematic study of statistical learning theory in CV quantum systems, leaving several open problems of interest, including: learnability of general CV circuits (see comment above); sample-complexity lower bounds; explicit learning algorithms for CV circuits and their application to specific tasks (e.g., the discovery of optimal decoders for optical communication~\cite{Notzel2022a%,RosatiUpcoming
,Rosati2016,Rosati2017a,Fanizza2020b} and state discrimination~\cite{Bilkis2020,Rosati17c} and the approximation of quantum processes with restricted resources~\cite{Diaz2020}); analysis of more challenging learning settings, where one is allowed to choose the quantum measurements to be applied on the samples and has to guarantee the approximation of multiple measurement operators at once~\cite{Aaronson2018,Huang2020b,Fanizza2022}.

\section{Acknowledgements}
M.R. acknowledges support from the Einstein Foundation. This project has received funding from the European Union's Horizon 2020 research and innovation programme under the Marie Sk\l odowska-Curie grant agreement No 845255 and the project PNRR - Finanziato dall'Unione europea - MISSIONE 4 COMPONENTE 2 INVESTIMENTO 1.2 - ``Finanziamento di progetti presentati da giovani ricercatori'' - Id MSCA\_0000011-SQUID - CUP F83C22002390007 (Young Researchers) - Finanziato dall'Unione europea - NextGenerationEU.

\appendix
{\color{blue}
\section{Agnostic learning of probabilistic concepts with point-wise guarantees} \label{app:proof_pfunction_learnability}
Here we prove a learning theorem for probabilistic concepts analogous to ~\cite[Supp. Mat. Theorem 1]{Aaronson2007}, but dropping the assumption that the unknown concept $f$ is contained in the hypothesis class $\cF$. We use a quadratic loss to measure the distance of a hypothesis $h\in\cF$ from the true concept on input $x$.
In this agnostic setting $f\notin\cF$, hence the aim is to ensure that, if one finds a hypothesis $h\in\cF$ with small loss on the training set, then $h$ generalizes well, i.e., its loss is small also on unseen inputs. For the proof we need the following theorem on function learning from interpolation by Anthony and Bartlett~\cite{Anthony2000}, which applies to the agnostic setting:
\begin{theorem}
    (\cite[Corollary 3.3]{Anthony2000})\label{thm:learning_from_interpolation}  Let $\cX$ be a sample space, $Q$ a distribution on $\cX$ and $\cF\subseteq\R^\cX$ a hypothesis class of real-valued functions. Consider an unknown concept $f:\cX\mapsto\R$ (not necessarily in $\cF$) and a training set $\{(x_i,f(x_i)\}_{i=1}^T$, where $x_i\sim Q$. Suppose that we choose a hypothesis $h\in\cF$ with small loss on each element of the training set, i.e., $|h(x_i)-f(x_i)|<\eta$ for all $i=1,\cdots,T$. Then with probability at least $1-\delta$ over the sample it holds
    \begin{equation}
{\rm Pr}_{x\sim Q}( |h(x)-f(x)|< \eta + \gamma)\geq 1-\epsilon,        
\end{equation}
i.e., the hypothesis has small loss also on unseen instances of the random variable $x$, provided that
\begin{equation}\label{eq:sample_complexity_learning_from_interpolation}
    T\geq T_0(\epsilon, \gamma, \delta) = O\left(\frac1\epsilon\left(d\log^2\frac{d}{\gamma \epsilon} + \log\frac{1}{\delta}\right)\right),
\end{equation}
where $d = {\rm fat}_{\cF}\left(\frac{\gamma}{8}\right)$, for all $\eta,\epsilon,\gamma,\delta>0$.
\end{theorem}
The key difference of this Theorem with respect to~\cite{Alon1997,Bartlett1998} is that it provides convergence guarantees for the point-wise loss, i.e., for all $x$, rather than for its expectation. This requirement is necessary for quantum state (measurement, and channel) learning à la Aaronson, where one wants to identify an operator that reproduces the observed statistics with small error on each input. Indeed, this problem corresponds to that of \emph{learning a model of probability} introduced by Kearns and Schapire~\cite{Kearns1994}. In contrast, the circuit learning tasks introduced in Sec.~\ref{sec:circuitTrain} only request a convergence guarantee for the expected loss, thus corresponding to the problem of \emph{learning a decision rule} also introduced by Kearns and Schapire~\cite{Kearns1994}.

We are now ready to state and prove our theorem, relying on a useful trick by~\cite{Aaronson2007} to learn from measurement outcomes rather than probability values.
\begin{theorem}(Theorem~\ref{thm:agnostic_pfunction_learning} repeated) 
Let $\cX$ be a sample space, $Q$ a distribution on $\cX$ and $\cF\subseteq[0,1]^\cX$ a hypothesis class. Consider an unknown p-concept $f:\cX\mapsto[0,1]$ (not necessarily in $\cF$) and a training set $\{(x_i,b_i)\}_{i=1}^T$, where $x_i\sim Q$ and $b_i=1$ with probability $f(x_i)$ or $b_i=0$ otherwise. Suppose that we choose a hypothesis $h\in\cF$ with small quadratic loss on each element of the training set, i.e., $(h(x_i)-b_i)^2<\eta$ for all $i=1,\cdots,T$. Then with probability at least $1-\delta$ over the sample it holds
\begin{equation}
    \Pr_{x\sim Q}((h(x)-f(x))^2<\eta+\gamma)\geq 1-\epsilon
\end{equation}
for all $\eta,\gamma,\epsilon,\delta>0$, provided that 
\begin{equation}
    T\geq T_0(\epsilon,\gamma,\delta,\cF) = O\left(\frac1\epsilon\left(d\log^2\frac{d}{\gamma \epsilon} + \log\frac{1}{\delta}\right)\right),
\end{equation}
where $d = {\rm fat}_{\cF}\left(\frac{\gamma}{8}\right)$.
\end{theorem}
\begin{proof}
    Following~\cite[Suppl. Mat. Theorem 1]{Aaronson2007}, starting from $\cF$ we construct a suitable class of functions that take as input both random variables $x\in \cX$ and $b\in[0,1]$ as follows:
    \begin{equation}
        \cF^* = \{h^*:(x,b)\mapsto (h(x)-b)^2\}.
    \end{equation}
    We can now apply Theorem~\ref{thm:learning_from_interpolation} to the sample space $\cX\times\{0,1\}$, with probability distribution $Q\cdot f(x)$ and true concept $f^*(x,b)=0$.  Suppose that a minimization algorithm finds a concept $h^*\in\cF$ with $\eta$-small loss on each element of the training set, i.e., for all $i=1,\cdots,T$ it holds
    \begin{equation}
        \eta>|h^*(x_i,b_i)-f^*(x_i,b_i)| = (h(x_i)-b_i)^2.
    \end{equation}
Then, if 
\begin{equation}
    T\geq T_0(\epsilon,\gamma,\delta,\cF^*) = O\left(\frac1\epsilon\left(d\log^2\frac{d}{\gamma \epsilon} + \log\frac{1}{\delta}\right)\right),
\end{equation}
with $d = {\rm fat}_{\cF^*}\left(\frac{\gamma}{8}\right)$, with probability at least $1-\delta$ over the sample it holds
\begin{align}
   1-\epsilon &\leq \Pr_{x\sim Q, b\sim f(x)}( |h^*(x,b)-f^*(x,b)|<\eta+\gamma)\\
   &= \Pr_{x\sim Q}((h(x)-B)^2<\eta+\gamma,B=b)\\
   &= \Pr_{x\sim Q}(\cup_{b\in\{0,1\}}(h(x)-b)^2<\eta+\gamma)\\
   &\leq \Pr_{x\sim Q}(\E{b\sim f(x)}{(h(x)-b)^2}<\eta+\gamma)\\
   &\leq \Pr_{x\sim Q}((h(x)-f(x))^2<\eta+\gamma),
\end{align}
 where the first inequality follows from Theorem \ref{thm:learning_from_interpolation}, the first equality from conditioning on the binary random variable $B$ with distribution $f(x)$, and the second equality from the fact that the events $B=0$ and $B=1$ are disjoint. Furthermore, the second and third inequalities follow from the fact that $(E_1\Rightarrow E_2)$ implies $\Pr(E_1)\leq\Pr(E_2)$ for any two events $E_1$, $E_2$. In particular, for the second inequality we observed that
 \begin{equation}
     (h(x)-b)^2<\eta+\gamma\; \forall b\in\{0,1\} \Rightarrow \E{b\sim f(x)}{(h(x)-b)^2}<\eta+\gamma,
 \end{equation}
while for the third one we have
 \begin{align}
     \eta+\gamma > \E{b\sim f(x)}{(h(x)-b)^2} &= f(x) (1-h(x))^2 + (1-f(x)) h(x)^2 \\
     &= (h(x)-f(x))^2 +f(x) - f(x)^2 \geq (h(x)-f(x))^2,
 \end{align}
 since $f(x)\in[0,1]$. Finally, we observe that $d = {\rm fat}_{\cF^*}\left(\alpha\right) \leq 2 {\rm fat}_{\cF}\left(\alpha\right)$, obtaining the desired sample complexity.
 \end{proof}

}

%\appendix
%\section{Details}

%\begin{figure}[t]
   % \centering
    %\includegraphics[width=0.5\textwidth]{UCB-QL.png}
    %\caption{We compare two different variants of UCB showing that the exploration-exploitation trade-off is an intrinsic feature of our problem.}
    %\label{fig:ucbs}
%\end{figure}

\bibliographystyle{quantum}
   \bibliography{library.bib}

\begin{thebibliography}{10}

\bibitem{Zhong2020}
Han-Sen Zhong, Hui Wang, Yu-Hao Deng, Ming-Cheng Chen, Li-Chao Peng, Yi-Han
  Luo, Jian Qin, Dian Wu, Xing Ding, Yi~Hu, Peng Hu, Xiao-Yan Yang, Wei-Jun
  Zhang, Hao Li, Yuxuan Li, Xiao Jiang, Lin Gan, Guangwen Yang, Lixing You,
  Zhen Wang, Li~Li, Nai-Le Liu, Chao-Yang Lu, and Jian-Wei Pan.
\newblock ``{Quantum computational advantage using photons}''.
\newblock \href{https://dx.doi.org/10.1126/science.abe8770}{Science (80-. ).
  {\bf 370}, 1460 LP -- 1463}~(2020).

\bibitem{Madsen2022}
Lars~S. Madsen, Fabian Laudenbach, Mohsen~Falamarzi. Askarani, Fabien Rortais,
  Trevor Vincent, Jacob F.~F. Bulmer, Filippo~M. Miatto, Leonhard Neuhaus,
  Lukas~G. Helt, Matthew~J. Collins, Adriana~E. Lita, Thomas Gerrits, Sae~Woo
  Nam, Varun~D. Vaidya, Matteo Menotti, Ish Dhand, Zachary Vernon,
  Nicol{\'{a}}s Quesada, and Jonathan Lavoie.
\newblock ``{Quantum computational advantage with a programmable photonic
  processor}''.
\newblock \href{https://dx.doi.org/10.1038/s41586-022-04725-x}{Nature {\bf
  606}, 75--81}~(2022).

\bibitem{Hoch2021}
Francesco Hoch, Simone Piacentini, Taira Giordani, Zhen-Nan Tian, Mariagrazia
  Iuliano, Chiara Esposito, Anita Camillini, Gonzalo Carvacho, Francesco
  Ceccarelli, Nicol{\`{o}} Spagnolo, Andrea Crespi, Fabio Sciarrino, and
  Roberto Osellame.
\newblock ``{Reconfigurable continuously-coupled 3D photonic circuit for Boson
  Sampling experiments}''.
\newblock \href{https://dx.doi.org/10.1038/s41534-022-00568-6}{npj Quantum Inf.
  {\bf 8}, 55}~(2022).
\newblock  \href{http://arxiv.org/abs/2106.08260}{arXiv:2106.08260}.

\bibitem{EliBourassa2021}
J.~{Eli Bourassa}, Rafael~N Alexander, Michael Vasmer, Ashlesha Patil, Ilan
  Tzitrin, Takaya Matsuura, Daiqin Su, Ben~Q Baragiola, Saikat Guha, Guillaume
  Dauphinais, Krishna~K Sabapathy, Nicolas~C Menicucci, and Ish Dhand.
\newblock ``{Blueprint for a scalable photonic fault-tolerant quantum
  computer}''.
\newblock \href{https://dx.doi.org/10.22331/Q-2021-02-04-392}{Quantum {\bf 5},
  1--38}~(2021).
\newblock  \href{http://arxiv.org/abs/2010.02905}{arXiv:2010.02905}.

\bibitem{Wang2020}
Jianwei Wang, Fabio Sciarrino, Anthony Laing, and Mark~G. Thompson.
\newblock ``{Integrated photonic quantum technologies}''.
\newblock \href{https://dx.doi.org/10.1038/s41566-019-0532-1}{Nat. Photonics
  {\bf 14}, 273--284}~(2020).

\bibitem{Krastanov2015}
Stefan Krastanov, Victor~V. Albert, Chao Shen, Chang-Ling Zou, Reinier~W.
  Heeres, Brian Vlastakis, Robert~J. Schoelkopf, and Liang Jiang.
\newblock ``{Universal Control of an Oscillator with Dispersive Coupling to a
  Qubit}''~(2015).
\newblock  \href{http://arxiv.org/abs/1502.08015}{arXiv:1502.08015}.

\bibitem{Eickbusch2021}
Alec Eickbusch, Volodymyr Sivak, Andy~Z. Ding, Salvatore~S. Elder, Shantanu~R.
  Jha, Jayameenakshi Venkatraman, Baptiste Royer, S.~M. Girvin, Robert~J.
  Schoelkopf, and Michel~H. Devoret.
\newblock ``{Fast Universal Control of an Oscillator with Weak Dispersive
  Coupling to a Qubit}''~(2021).
\newblock  \href{http://arxiv.org/abs/2111.06414}{arXiv:2111.06414}.

\bibitem{Arrangoiz-Arriola2019}
Patricio Arrangoiz-Arriola, E.~Alex Wollack, Zhaoyou Wang, Marek Pechal, Wentao
  Jiang, Timothy~P. McKenna, Jeremy~D. Witmer, and Amir~H. Safavi-Naeini.
\newblock ``{Resolving the energy levels of a nanomechanical
  oscillator}''~(2019).
\newblock  \href{http://arxiv.org/abs/1902.04681}{arXiv:1902.04681}.

\bibitem{Andersen2014a}
Ulrik~L. Andersen, Jonas~S. Neergaard-Nielsen, Peter {Van Loock}, and Akira
  Furusawa.
\newblock ``{Hybrid discrete- and continuous-variable quantum information}''.
\newblock \href{https://dx.doi.org/10.1038/nphys3410}{Nat. Phys. {\bf 11},
  713--719}~(2015).
\newblock  \href{http://arxiv.org/abs/1409.3719}{arXiv:1409.3719}.

\bibitem{Hastrup2019}
Jacob Hastrup, Kimin Park, Jonatan~Bohr Brask, Radim Filip, and Ulrik~Lund
  Andersen.
\newblock ``{Measurement-free preparation of grid states}''.
\newblock \href{https://dx.doi.org/10.1038/s41534-020-00353-3}{npj Quantum Inf.
  {\bf 7}, 17}~(2021).

\bibitem{Michael2016}
Marios~H. Michael, Matti Silveri, R.~T. Brierley, Victor~V. Albert, Juha
  Salmilehto, Liang Jiang, and S.~M. Girvin.
\newblock ``{New class of quantum error-correcting codes for a bosonic mode}''.
\newblock \href{https://dx.doi.org/10.1103/PhysRevX.6.031006}{Phys. Rev. X{\bf
  6}}~(2016).
\newblock  \href{http://arxiv.org/abs/1602.00008}{arXiv:1602.00008}.

\bibitem{Noh2020}
Kyungjoo Noh, S~M Girvin, and Liang Jiang.
\newblock ``{Encoding an Oscillator into Many Oscillators}''.
\newblock \href{https://dx.doi.org/10.1103/PhysRevLett.125.080503}{Phys. Rev.
  Lett.{\bf 125}}~(2020).
\newblock  \href{http://arxiv.org/abs/1903.12615}{arXiv:1903.12615}.

\bibitem{Tosta2022a}
Allan D.~C. Tosta, Thiago~O. Maciel, and Leandro Aolita.
\newblock ``{Grand Unification of continuous-variable codes}''~(2022).
\newblock  \href{http://arxiv.org/abs/2206.01751}{arXiv:2206.01751}.

\bibitem{Aaronson2007}
Scott Aaronson.
\newblock ``{The Learnability of Quantum States}''.
\newblock \href{https://dx.doi.org/10.1098/rspa.2007.0113}{Proc. R. Soc. A
  Math. Phys. Eng. Sci. {\bf 463}, 3089--3114}~(2006).
\newblock  \href{http://arxiv.org/abs/0608142}{arXiv:0608142}.

\bibitem{Cheng2015}
Hao-Chung Cheng, Min-Hsiu Hsieh, and Ping-Cheng Yeh.
\newblock ``{The Learnability of Unknown Quantum Measurements}''.
\newblock Quantum Inf. Comput. {\bf 16}, 0615--0656~(2016).
\newblock  \href{http://arxiv.org/abs/1501.00559}{arXiv:1501.00559}.

\bibitem{Sweke2020a}
Ryan Sweke, Jean~Pierre Seifert, Dominik Hangleiter, and Jens Eisert.
\newblock ``{On the quantum versus classical learnability of discrete
  distributions}''.
\newblock \href{https://dx.doi.org/10.22331/Q-2021-03-23-417}{Quantum{\bf
  5}}~(2021).
\newblock  \href{http://arxiv.org/abs/2007.14451}{arXiv:2007.14451}.

\bibitem{Caro2020a}
Matthias~C. Caro and Ishaun Datta.
\newblock ``{Pseudo-dimension of quantum circuits}''.
\newblock \href{https://dx.doi.org/10.1007/s42484-020-00027-5}{Quantum Mach.
  Intell. {\bf 2}, 14}~(2020).
\newblock  \href{http://arxiv.org/abs/2002.01490}{arXiv:2002.01490}.

\bibitem{Caro2021b}
Matthias~C. Caro, Elies Gil-Fuster, Johannes~Jakob Meyer, Jens Eisert, and Ryan
  Sweke.
\newblock ``{Encoding-dependent generalization bounds for parametrized quantum
  circuits}''.
\newblock \href{https://dx.doi.org/10.22331/q-2021-11-17-582}{Quantum {\bf 5},
  582}~(2021).

\bibitem{Vapnik1999}
V.N. Vapnik.
\newblock ``{An overview of statistical learning theory}''.
\newblock \href{https://dx.doi.org/10.1109/72.788640}{IEEE Trans. Neural
  Networks {\bf 10}, 988--999}~(1999).

\bibitem{Vapnik2000}
Vladimir~N. Vapnik.
\newblock ``{The Nature of Statistical Learning Theory}''.
\newblock \href{https://dx.doi.org/10.1007/978-1-4757-3264-1}{Volume~13}.
\newblock Springer New York. New York, NY~(2000).

\bibitem{Anthony1999}
Martin Anthony and Peter~L. Bartlett.
\newblock ``{Neural Network Learning}''.
\newblock \href{https://dx.doi.org/10.1017/cbo9780511624216}{Cambridge
  University Press}. ~(1999).

\bibitem{Popescu2021}
Claudiu~Marius Popescu.
\newblock ``{Learning bounds for quantum circuits in the agnostic setting}''.
\newblock \href{https://dx.doi.org/10.1007/s11128-021-03225-7}{Quantum Inf.
  Process. {\bf 20}, 286}~(2021).

\bibitem{Arunachalam2021}
Srinivasan Arunachalam, Yihui Quek, and John Smolin.
\newblock ``{Private learning implies quantum stability}''~(2021).
\newblock  \href{http://arxiv.org/abs/2102.07171}{arXiv:2102.07171}.

\bibitem{Gandhari2022}
Srilekha Gandhari, Victor~V. Albert, Thomas Gerrits, Jacob~M. Taylor, and
  Michael~J. Gullans.
\newblock ``{Continuous-Variable Shadow Tomography}''~(2022).
\newblock  \href{http://arxiv.org/abs/2211.05149}{arXiv:2211.05149}.

\bibitem{Becker2022}
Simon Becker, Nilanjana Datta, Ludovico Lami, and Cambyse Rouz{\'{e}}.
\newblock ``{Classical shadow tomography for continuous variables quantum
  systems}''~(2022).
\newblock  \href{http://arxiv.org/abs/2211.07578}{arXiv:2211.07578}.

\bibitem{Aaronson2018}
Scott Aaronson.
\newblock ``{Shadow tomography of quantum states}''.
\newblock In Proc. 50th Annu. ACM SIGACT Symp. Theory Comput.
\newblock \href{https://dx.doi.org/10.1145/3188745.3188802}{Pages 325--338}.
\newblock New York, NY, USA~(2018). ACM.
\newblock  \href{http://arxiv.org/abs/1711.01053}{arXiv:1711.01053}.

\bibitem{Volkoff2021b}
Tyler Volkoff, Zo{\"{e}} Holmes, and Andrew Sornborger.
\newblock ``{Universal Compiling and (No-)Free-Lunch Theorems for
  Continuous-Variable Quantum Learning}''.
\newblock \href{https://dx.doi.org/10.1103/PRXQuantum.2.040327}{PRX Quantum
  {\bf 2}, 040327}~(2021).

\bibitem{Lloyd1999}
Seth Lloyd and Samuel~L. Braunstein.
\newblock ``{Quantum Computation over Continuous Variables}''.
\newblock \href{https://dx.doi.org/10.1103/PhysRevLett.82.1784}{Phys. Rev.
  Lett. {\bf 82}, 1784--1787}~(1999).

\bibitem{Bourassa2021}
J.~Eli Bourassa, Nicol{\'{a}}s Quesada, Ilan Tzitrin, Antal Sz{\'{a}}va,
  Theodor Isacsson, Josh Izaac, Krishna~Kumar Sabapathy, Guillaume Dauphinais,
  and Ish Dhand.
\newblock ``{Fast Simulation of Bosonic Qubits via Gaussian Functions in Phase
  Space}''.
\newblock \href{https://dx.doi.org/10.1103/PRXQuantum.2.040315}{PRX Quantum
  {\bf 2}, 040315}~(2021).
\newblock  \href{http://arxiv.org/abs/2103.05530}{arXiv:2103.05530}.

\bibitem{Rosati2017}
Matteo Rosati, Andrea Mari, and Vittorio Giovannetti.
\newblock ``{Capacity of coherent-state adaptive decoders with interferometry
  and single-mode detectors}''.
\newblock \href{https://dx.doi.org/10.1103/PhysRevA.96.012317}{Phys. Rev. A
  {\bf 96}, 012317}~(2017).
\newblock  \href{http://arxiv.org/abs/1703.05701}{arXiv:1703.05701}.

\bibitem{Fanizza2020b}
Marco Fanizza, Matteo Rosati, Michalis Skotiniotis, John Calsamiglia, and
  Vittorio Giovannetti.
\newblock ``{Squeezing-enhanced communication without a phase reference}''.
\newblock \href{https://dx.doi.org/10.22331/q-2021-12-23-608}{Quantum {\bf 5},
  608}~(2021).
\newblock  \href{http://arxiv.org/abs/2006.06522}{arXiv:2006.06522}.

\bibitem{Bilkis2021a}
M.~Bilkis, Matteo Rosati, and John Calsamiglia.
\newblock ``{Reinforcement-learning calibration of coherent-state receivers on
  variable-loss optical channels}''.
\newblock In 2021 IEEE Inf. Theory Work.
\newblock \href{https://dx.doi.org/10.1109/ITW48936.2021.9611396}{Pages 1--6}.
\newblock IEEE~(2021).

\bibitem{Diaz2020}
Maria~Garcia Diaz, Benjamin Desef, Matteo Rosati, Dario Egloff, John
  Calsamiglia, Andrea Smirne, Michalis Skotiniotis, and Susana~F. Huelga.
\newblock ``{Accessible coherence in open quantum system dynamics}''.
\newblock \href{https://dx.doi.org/10.22331/q-2020-04-02-249}{Quantum{\bf
  4}}~(2020).
\newblock  \href{http://arxiv.org/abs/1910.05089}{arXiv:1910.05089}.

\bibitem{Nakahira2021a}
Kenji Nakahira.
\newblock ``{Quantum process discrimination with restricted strategies}''.
\newblock \href{https://dx.doi.org/10.1103/PhysRevA.104.062609}{Phys. Rev. A
  {\bf 104}, 062609}~(2021).
\newblock  \href{http://arxiv.org/abs/2104.09038}{arXiv:2104.09038}.

\bibitem{Nakahira2021}
Kenji Nakahira and Kentaro Kato.
\newblock ``{Generalized quantum process discrimination problems}''.
\newblock \href{https://dx.doi.org/10.1103/PhysRevA.103.062606}{Phys. Rev. A
  {\bf 103}, 062606}~(2021).
\newblock  \href{http://arxiv.org/abs/2104.09759}{arXiv:2104.09759}.

\bibitem{Sidhu2021a}
Jasminder~S. Sidhu, Michael~S. Bullock, Saikat Guha, and Cosmo Lupo.
\newblock ``{Unambiguous discrimination of coherent states}''~(2021).
\newblock  \href{http://arxiv.org/abs/2109.00008}{arXiv:2109.00008}.

\bibitem{Becerra2013a}
F~E Becerra, J~Fan, G~Baumgartner, J~Goldhar, J~T Kosloski, and a~Migdall.
\newblock ``{Experimental demonstration of a receiver beating the standard
  quantum limit for multiple nonorthogonal state discrimination}''.
\newblock \href{https://dx.doi.org/10.1038/nphoton.2012.316}{Nat. Photonics
  {\bf 7}, 147--152}~(2013).

\bibitem{DiMario2022}
M.~T. DiMario and F.~E. Becerra.
\newblock ``{Demonstration of optimal non-projective measurement of binary
  coherent states with photon counting}''.
\newblock \href{https://dx.doi.org/10.1038/s41534-022-00595-3}{npj Quantum Inf.
  {\bf 8}, 84}~(2022).

\bibitem{Rosati17c}
Matteo Rosati, Giacomo {De Palma}, Andrea Mari, and Vittorio Giovannetti.
\newblock ``{Optimal quantum state discrimination via nested binary
  measurements}''.
\newblock \href{https://dx.doi.org/10.1103/PhysRevA.95.042307}{Phys. Rev. A -
  At. Mol. Opt. Phys. {\bf 95}, 1--10}~(2017).
\newblock  \href{http://arxiv.org/abs/1701.02233}{arXiv:1701.02233}.

\bibitem{Assalini2011}
Antonio Assalini, Nicola {Dalla Pozza}, and Gianfranco Pierobon.
\newblock ``{Revisiting the Dolinar receiver through multiple-copy state
  discrimination theory}''.
\newblock \href{https://dx.doi.org/10.1103/PhysRevA.84.022342}{Phys. Rev. A -
  At. Mol. Opt. Phys. {\bf 84}, 1--6}~(2011).
\newblock  \href{http://arxiv.org/abs/1107.5452}{arXiv:1107.5452}.

\bibitem{Rahman2017}
Sharif Rahman.
\newblock ``{Wiener-Hermite Polynomial Expansion for Multivariate Gaussian
  Probability Measures}''~(2017).
\newblock  \href{http://arxiv.org/abs/1704.07912}{arXiv:1704.07912}.

\bibitem{Dodonov1994}
V.~V. Dodonov, O.~V. Manko, and V.~I. Manko.
\newblock ``{Multidimensional Hermite polynomials and photon distribution for
  polymode mixed light}''.
\newblock \href{https://dx.doi.org/10.1103/PhysRevA.50.813}{Phys. Rev. A {\bf
  50}, 813--817}~(1994).

\bibitem{Kearns1994}
Michael~J. Kearns and Robert~E. Schapire.
\newblock ``{Efficient distribution-free learning of probabilistic concepts}''.
\newblock \href{https://dx.doi.org/10.1016/S0022-0000(05)80062-5}{J. Comput.
  Syst. Sci. {\bf 48}, 464--497}~(1994).

\bibitem{Alon1997}
Noga Alon, Shai Ben-David, Nicol{\`{o}} Cesa-Bianchi, and David Haussler.
\newblock ``{Scale-sensitive dimensions, uniform convergence, and
  learnability}''.
\newblock \href{https://dx.doi.org/10.1145/263867.263927}{J. ACM {\bf 44},
  615--631}~(1997).

\bibitem{Bartlett1998}
Peter~L. Bartlett and Philip~M. Long.
\newblock ``{Prediction, Learning, Uniform Convergence, and Scale-Sensitive
  Dimensions}''.
\newblock \href{https://dx.doi.org/10.1006/jcss.1997.1557}{J. Comput. Syst.
  Sci. {\bf 56}, 174--190}~(1998).

\bibitem{Anthony2000}
Martin Anthony and Peter~L. Bartlett.
\newblock ``{Function Learning from Interpolation}''.
\newblock \href{https://dx.doi.org/10.1017/S0963548300004247}{Comb. Probab.
  Comput. {\bf 9}, 213--225}~(2000).

\bibitem{Haussler1992}
David Haussler.
\newblock ``{Decision theoretic generalizations of the PAC model for neural net
  and other learning applications}''.
\newblock \href{https://dx.doi.org/10.1016/0890-5401(92)90010-D}{Inf. Comput.
  {\bf 100}, 78--150}~(1992).

\bibitem{Asor2014}
Ohad Asor, Hubert~Haoyang Duan, and Aryeh Kontorovich.
\newblock ``{On the additive properties of the fat-shattering dimension}''.
\newblock \href{https://dx.doi.org/10.1109/TNNLS.2014.2327065}{IEEE Trans.
  Neural Networks Learn. Syst. {\bf 25}, 2309--2312}~(2014).

\bibitem{Freund1997}
Yoav Freund and Robert~E Schapire.
\newblock ``{A Decision-Theoretic Generalization of On-Line Learning and an
  Application to Boosting}''.
\newblock \href{https://dx.doi.org/https://doi.org/10.1006/jcss.1997.1504}{J.
  Comput. Syst. Sci. {\bf 55}, 119--139}~(1997).

\bibitem{serafiniBOOK}
Alessio Serafini.
\newblock ``{Quantum Continuous Variables}''.
\newblock \href{https://dx.doi.org/10.1201/9781315118727}{CRC Press}. Boca
  Raton, FL : CRC Press, Taylor {\&} Francis Group, [2017] |~(2017).

\bibitem{Notzel2022a}
Janis N{\"{o}}tzel and Matteo Rosati.
\newblock ``{Operating Fiber Networks in the Quantum Limit}''~(2022).
\newblock  \href{http://arxiv.org/abs/2201.12397}{arXiv:2201.12397}.

\bibitem{Rosati2016}
Matteo Rosati and Vittorio Giovannetti.
\newblock ``{Achieving the Holevo bound via a bisection decoding protocol}''.
\newblock \href{https://dx.doi.org/10.1063/1.4953690}{J. Math. Phys.{\bf
  57}}~(2016).
\newblock  \href{http://arxiv.org/abs/1506.04999}{arXiv:1506.04999}.

\bibitem{Rosati2017a}
Matteo Rosati, Andrea Mari, and Vittorio Giovannetti.
\newblock ``{Capacity of coherent-state adaptive decoders with interferometry
  and single-mode detectors}''.
\newblock \href{https://dx.doi.org/10.1103/PhysRevA.96.012317}{Phys. Rev. A
  {\bf 96}, 012317}~(2017).

\bibitem{Bilkis2020}
M.~Bilkis, M.~Rosati, R.~Morral Yepes, and J.~Calsamiglia.
\newblock ``{Real-time calibration of coherent-state receivers: Learning by
  trial and error}''.
\newblock \href{https://dx.doi.org/10.1103/PhysRevResearch.2.033295}{Phys. Rev.
  Res. {\bf 2}, 033295}~(2020).
\newblock  \href{http://arxiv.org/abs/2001.10283}{arXiv:2001.10283}.

\bibitem{Huang2020b}
Hsin-Yuan Huang, Richard Kueng, and John Preskill.
\newblock ``{Predicting many properties of a quantum system from very few
  measurements}''.
\newblock \href{https://dx.doi.org/10.1038/s41567-020-0932-7}{Nat. Phys. {\bf
  16}, 1050--1057}~(2020).

\bibitem{Fanizza2022}
Marco Fanizza, Yihui Quek, and Matteo Rosati.
\newblock ``{Learning quantum processes without input control}''~(2022).
\newblock  \href{http://arxiv.org/abs/2211.05005}{arXiv:2211.05005}.

\end{thebibliography}

\end{document}